\documentclass[ASNA,twocolumn]{USG} 
\usepackage{anyfontsize} %
\usepackage{colortbl}   
\usepackage{xcolor}     
\usepackage{steinmetz}
\usepackage{textcomp}
\usepackage{comment}
\usepackage{mathtools}
\graphicspath{{./images/}}
\articletype{RESEARCH ARTICLE}%

\begin{document}
\title{Scenario-based Data-Enabled Predictive Control: \\Robustification via the Scenario Approach}
\transtitle{Scenario-based Data-Enabled Predictive Control}
\author[1]{Sebastian Zieglmeier}[https://orcid.org/0009-0009-2516-3723]
\author[1]{Nikolas Recke}[https://orcid.org/0009-0007-3362-8194]
\author[1]{Mathias Hudoba de Badyn}[https://orcid.org/0000-0003-0955-2381]

\authormark{Zieglmeier \textsc{et al.}}
\titlemark{Scenario-based Data-Enabled Predictive Control: Robustification via the Scenario Approach}

\address[1]{\orgdiv{Department of Technology Systems, }\orgname{University of Oslo, }%
\orgaddress{\state{Oslo, }\country{Norway}}}%

\corres{Sebastian Zieglmeier  (\email{sebastiz@uio.no}) ~|~ Nikolas Recke (\email{nikolalr@uio.no}) ~|~ Mathias Hudoba de Badyn (\email{mathihud@uio.no})}



\fundingInfo{IDES (Intelligent Dynamic Energy Systems) project; FME Solar, Grant/Award Number: 350244}


\transkeywords{data-enabled predictive control | scenario approach | probabilistic constraint satisfaction | data-driven control | chance-constrained control}

\abstract[ABSTRACT]{This paper proposes Scenario-Based Data-Enabled Predictive Control (Scenario-DeePC), which integrates the scenario optimization framework into Data-enabled Predictive Control (DeePC) to provide probabilistic guarantees on constraint satisfaction under uncertainty. In contrast to existing methods, the uncertainty is characterized directly from data by constructing empirical disturbance scenarios from observed prediction errors, keeping the method fully consistent with the data-driven philosophy of DeePC and free of distributional assumptions. 
We establish the supporting theory, including a distribution-free probabilistic guarantee on constraint satisfaction and recursive feasibility of the receding-horizon scheme.
An adaptive extension collects scenarios online, enabling the controller to adjust to changing noise characteristics, disturbances, and operating-point-dependent model mismatches.
The approach is demonstrated on a linear Boeing 747 model and a nonlinear two-tank system, showing a significant reduction in constraint violations compared to standard DeePC, while maintaining comparable tracking performance in nominal conditions and improving tracking accuracy in the nonlinear setting.}



\copyright{This work has been submitted for possible publication. Copyright may be transferred without notice, after which this version may no longer be accessible.}


\maketitle


\section{Introduction}
\label{sec:Introduction}
Data-enabled Predictive Control, first proposed in~\cite{Coulson_2019}, provides a purely data-driven alternative to traditional Model Predictive Control (MPC). Rather than constructing and relying on an explicit parametric model of the system dynamics, DeePC directly formulates the optimal control problem using measured input-output data, avoiding an intermediate modeling step and the errors it can introduce. 
DeePC relies on Willems’ fundamental lemma~\cite{Willems_2004}, according to which measured trajectories of a persistently excited deterministic linear time-invariant (LTI) system span the entire behavior of the system. Organizing these trajectories in Hankel matrices provides a data-based representation of admissible future evolutions conditioned on past input-output sequences. This behavioral representation replaces the classical model-based predictor and removes the need for explicit knowledge of the system equations.
Real-world systems, however, are typically affected by measurement noise and may exhibit nonlinear behavior. As demonstrated in~\cite{Coulson_2019, Elokda_2021, Zieglmeier_2025, Huang_2023, Zieglmeier_2025_GS}, these challenges can be addressed by incorporating suitable regularization terms into the DeePC optimization problem. Such regularization relaxes the exact trajectory matching conditions implied by the lemma and enhances robustness with respect to noisy data and mild nonlinearities. 
Nevertheless, while regularization improves robustness of the implicit data-driven identification step in the trajectory reconstruction, as reviewed in~\cite[Sec.~3.2]{Berberich_2024}, it does not explicitly address the risk of constraint violations arising from uncertainty, like measurement noise and disturbances. 
Therefore, approaches explicitly addressing robust constraint satisfaction, ensuring performance and safety simultaneously, remain an open research question within data-driven predictive control. 

A principled framework for constraint satisfaction under uncertainty is provided by the scenario approach~\cite{Calafiore_2006, Campi_2008, Campi_2009}, which replaces intractable chance-constrained formulations with a finite set of sampled uncertainty realizations. For convex programs, this yields explicit, distribution-free probabilistic guarantees on constraint satisfaction without requiring knowledge of the disturbance distribution beyond the availability of independent samples. The scenario approach has been successfully integrated with MPC~\cite{Schildbach_2014, Calafiore_2012}, making its extension to the data-driven counterpart DeePC a natural step that enables systematic constraint satisfaction under noise and disturbances, complementing the regularization mechanisms of DeePC.

Despite these advances, a systematic integration of the scenario approach within the DeePC framework has, to the best of the authors’ knowledge, not yet been explored. Such a combination is appealing for several reasons. First, the scenario approach provides non-asymptotic probabilistic guarantees under the assumption of independent samples from an underlying uncertainty distribution, without requiring restrictive structural assumptions such as Gaussianity or boundedness. This makes it conceptually compatible with the data-driven philosophy of DeePC. Second, the resulting formulation remains computationally tractable, as uncertainty is incorporated through sampled constraints rather than through ambiguity sets or distributional reformulations. Third, the scenario methodology provides explicit design guidelines that relate the number of scenarios to the desired level of probabilistic robustness and the effective decision dimension. Finally, while regularized DeePC has demonstrated strong empirical and theoretical performance in terms of robust trajectory prediction and implicit data-driven identification under stochastic disturbances, the explicit enforcement of constraint satisfaction under uncertainty has received less attention. The proposed Scenario-DeePC framework, therefore, shifts the focus from enhancing prediction robustness to systematically enhancing probabilistic constraint satisfaction, thereby complementing existing regularization-based approaches rather than replacing them.

The main contributions of this paper are as follows:
\begin{itemize}
    \item \textbf{Method:} We propose Scenario-DeePC, which integrates the scenario approach into DeePC by constructing the uncertainty description directly from data: empirical scenarios are formed from the controller's own closed-loop prediction errors, keeping the method distribution-free and consistent with the data-driven philosophy of DeePC. An adaptive variant collects scenarios online, so the description tracks changing noise, disturbances, and operating-point-dependent model mismatch.
    \item \textbf{Theory:} We establish the supporting theory, including a distribution-free probabilistic guarantee on constraint satisfaction, recursive feasibility of the receding-horizon scheme, and an offset-correction property that compensates systematic prediction bias.
    \item \textbf{Validation:} We validate Scenario-DeePC on a linear Boeing~747 model and a nonlinear two-tank system, where it significantly reduces constraint violations relative to standard DeePC while retaining comparable nominal tracking for the Boeing~747 and improved tracking accuracy in the nonlinear case.
\end{itemize}
The remainder of the paper is organized as follows:
Section~\ref{sec:Problem} outlines the problem statement. 
Sections~\ref{sec:DeePC} and~\ref{sec:Scenario} provide a brief overview of DeePC and the scenario approach to introduce the notation, respectively. 
Section~\ref{sec:Scenario-DeePC} presents the proposed Scenario-DeePC approach and discusses its theoretical results in Section~\ref{sec:Theo_Results}.
Section~\ref{sec:Num_Results} showcases the numerical results, and Section~\ref{sec:Conclusion} concludes the paper.

\section{Problem statement}
\label{sec:Problem}
Consider a discrete-time LTI system affected by disturbances and measurement noise
\begin{equation}
\begin{aligned}
x_{k+1} &= A x_k + B u_k \\
y_k &= C x_k + D u_k + d_k + v_k ,
\end{aligned}
\label{eq:system}
\end{equation}
where $x_k \in \mathbb{R}^{n_x}$ denotes the system state, $u_k \in \mathbb{R}^{n_u}$ the control input, and $y_k \in \mathbb{R}^{n_y}$ the measured output at time step $k$. The variables $d_k \in \mathbb{R}^{n_y}$ and $v_k \in \mathbb{R}^{n_y}$ represent the disturbance acting on the output and measurement noise, respectively.
The matrices $A \in \mathbb{R}^{n_x \times n_x}$, $B \in \mathbb{R}^{n_x \times n_u}$, $C \in \mathbb{R}^{n_y \times n_x}$, $D \in \mathbb{R}^{n_y \times n_u}$ describe the system dynamics in state space representation.

\section{Data-Enabled Predictive Control}
\label{sec:DeePC}
We briefly recall the DeePC formulation to fix the notation used in the remainder of the paper. DeePC constructs predictions directly from previously collected input-output trajectories via the behavioral representation of the system, rather than from a parametric model such as~\eqref{eq:system}.
Formally, a dynamical system $\Sigma$ is defined as a triple $\Sigma = (\mathbb{T}, \mathbb{W}, \mathcal{B})$, where $\mathbb{T} \subseteq \mathbb{R}$ denotes the time axis, $\mathbb{W}$ the signal space, and $\mathcal{B} \subseteq \mathbb{W}^{\mathbb{T}}$ the behavior of the system. The behavior $\mathcal{B}$ represents the set of all trajectories that are compatible with the system dynamics, where $\mathbb{W}^{\mathbb{T}}$ denotes the set of all mappings from the time axis $\mathbb{T}$ to the signal space $\mathbb{W}$.
For input-output systems, the signal space is typically given by $\mathbb{W} = \mathbb{U} \times \mathbb{Y}$, such that the behavior consists of all admissible input-output trajectories $(u,y)$. For deterministic LTI systems, Willems' fundamental lemma~\cite{Willems_2004} states that every trajectory of the system can be expressed as a linear combination of segments of a single measured trajectory, provided that the applied input signal is persistently exciting of sufficiently high order. Consequently, a sufficiently rich measured trajectory spans the entire behavior of the system.
To exploit this result, a previously measured input-output trajectory
\begin{equation}
u^d = \{u_1^d, \dots, u_{T_d}^d\}, \qquad
y^d = \{y_1^d, \dots, y_{T_d}^d\}
\label{eq:data_trajectory}
\end{equation}
is arranged into Hankel matrices of depth $L$
\begin{equation}
\mathcal{H}_L(u^d) =
\begin{bmatrix}
u_1^d & u_2^d & \dots & u_{T_d-L+1}^d \\
u_2^d & u_3^d & \dots & u_{T_d-L+2}^d \\
\vdots & \vdots & \ddots & \vdots \\
u_L^d & u_{L+1}^d & \dots & u_{T_d}^d
\end{bmatrix},
\label{eq:hankel_u}
\end{equation}
and analogously for $\mathcal{H}_L(y^d)$. 
In particular, the input sequence $u^d$ is said to be persistently exciting of order $L+n_x$ if the corresponding Hankel matrix satisfies~\eqref{eq:pe_rank}~\cite{Coulson_2019}.
\begin{equation}
\operatorname{rank}\big(\mathcal{H}_L(u^d)\big) = n_u L.
\label{eq:pe_rank}
\end{equation}
This guarantees that the collected trajectory spans all admissible input-output trajectories of length $L$, where $L= T_{\mathrm{ini}} + N$ is divided into a past horizon $T_{\text{ini}}$ and a future horizon $N$.
Accordingly, the Hankel matrices are partitioned as
\begin{equation}
\begin{bmatrix}
U_p \\
U_f
\end{bmatrix}
=
\mathcal{H}_L(u^d),
\qquad
\begin{bmatrix}
Y_p \\
Y_f
\end{bmatrix}
=
\mathcal{H}_L(y^d),
\label{eq:hankel_partition}
\end{equation}
where $U_p$ and $Y_p$ contain the past input-output trajectories of length $T_{\text{ini}}$, while $U_f$ and $Y_f$ describe the corresponding future trajectories over the prediction horizon $N$.
At each control step, the most recent input-output measurements $u_{\text{ini}}$ and $y_{\text{ini}}$ are used to select trajectories in the data that are consistent with the current system state.
The DeePC control problem is formulated as
\begin{subequations}\label{eq:DeePC}
\begin{align}
\min_{g,\,u,\,y,\,\sigma_y}\quad
& \|r - y\|_Q^2 + \|u\|_R^2 + \|g\|_{\lambda_g}^2 + \|\sigma_y\|_{\lambda_y}^2
\label{eq:DeePC_cost}\\
\text{s.t.}\quad
& \left(\begin{array}{c} U_p \\ Y_p \\ U_f \\ Y_f \end{array}\right) g
= \left(\begin{array}{c} u_{\mathrm{ini}} \\ y_{\mathrm{ini}} \\ u \\ y \end{array}\right)
+ \left(\begin{array}{c} 0 \\ \sigma_y \\ 0 \\ 0 \end{array}\right)
\label{eq:DeePC_data}\\
& u_k \in \mathcal{U}, \; y_k \in \mathcal{Y}, \quad \forall k \in \{0, \ldots, N-1\}.
\label{eq:DeePC_constraints}
\end{align}
\end{subequations}
Here, the vector $g \in \mathbb{R}^{T_d - T_{\mathrm{ini}} - N + 1}$ parameterizes admissible trajectories by expressing them as linear combinations of columns of the data Hankel matrices, while $r \in \mathbb{R}^{n_y N}$ denotes the desired reference trajectory over the prediction horizon. The matrices $Q \succ 0$ and $R \succ 0$ are weighting matrices that penalize tracking errors and control effort, respectively.
To account for noise and model mismatches arising from measurement errors or mild nonlinear system behavior, regularization terms are incorporated into the DeePC optimization problem. In particular, the penalties $\| g \|_{\lambda_g}^2$ and $\| \sigma_y \|_{\lambda_y}^2$ relax the strict trajectory consistency imposed by Willems' fundamental lemma and improve robustness with respect to noisy data~\cite{Coulson_2019, Huang_2023}. 
Input and output constraints are enforced through the admissible sets $\mathcal{U}$ and $\mathcal{Y}$, which define the feasible operating region of the system. Explicit theoretical guarantees on robustness and closed-loop stability of DeePC have been established in~\cite{Berberich_2020, Berberich_2024}.

\section{Scenario Approach}
\label{sec:Scenario}
The scenario approach, introduced by~\cite{Calafiore_2006} and further developed in~\cite{Campi_2008}, provides a principled framework for solving convex optimization problems subject to uncertain constraints. Consider a generic chance-constrained program of the form
\begin{equation}
\begin{aligned}
    &\min_{x} \; c(x) \quad \\
    &\text{s.t.} \quad 
    \mathbb{P}\bigl\{f(x, \delta) \leq 0\bigr\} \geq 1 - \varepsilon,
    \label{eq:chance}
\end{aligned}
\end{equation}
where the decision variable $x \in \mathcal{X} \subseteq \mathbb{R}^{n_{\text{opt}}}$ ranges over a convex feasible set $\mathcal{X}$, $c(x)$ is a convex cost function, $\delta \in \Delta$ is an uncertain parameter drawn from a probability distribution $\mathbb{P}$, $\varepsilon \in (0,1)$ is a prescribed violation tolerance, and $n_{\text{opt}}$ is the number of decision variables. Despite its natural formulation, problem~\eqref{eq:chance} is generally intractable, which motivates the use of scenario-based approximations~\cite{Campi_2008}.
The scenario approach addresses this intractability by replacing the probabilistic constraint with a finite set of $N_{\mathrm{scen}}$ independently sampled constraint instances, called scenarios. Given independent and identically distributed (i.i.d.)\ samples $\{\delta^{(i)}\}_{i=1}^{N_{\mathrm{scen}}} \sim \mathbb{P}$, the scenario program reads
\begin{equation}
\begin{aligned}
    &\min_{x} \; c(x) \quad \\
    & \text{s.t.} \quad
    f\bigl(x, \delta^{(i)}\bigr) \leq 0, \quad i = 1, \ldots, N_{\mathrm{scen}}.
    \label{eq:scenario}
\end{aligned}
\end{equation}
Problem~\eqref{eq:scenario} is a standard convex program solvable with off-the-shelf solvers. Its solution $x_{N_{\mathrm{scen}}}^\star$ carries rigorous probabilistic guarantees on constraint satisfaction for unseen realizations of $\delta$.
The fundamental result of \cite{Campi_2008} characterizes the violation probability 
\begin{equation}
    V(x_{N_{\mathrm{scen}}}^\star) = \mathbb{P}\{\delta : f(x_{N_{\mathrm{scen}}}^\star, \delta) > 0\}.
\end{equation}
For a confidence parameter $\beta \in (0,1)$, it holds that
\begin{equation}
    \mathbb{P}^{N_{\mathrm{scen}}}\bigl\{V(x_{N_{\mathrm{scen}}}^\star) > \varepsilon\bigr\} \leq 
    \sum_{i=0}^{n_{\text{opt}}-1} \binom{{N_{\mathrm{scen}}}}{i} \varepsilon^i (1-\varepsilon)^{{N_{\mathrm{scen}}}-i},
    \label{eq:scenario_bound}
\end{equation}
where $n_{\mathrm{opt}}$ is the number of optimization variables and $\mathbb{P}^{N_{\mathrm{scen}}}$ denotes the product measure induced by the i.i.d. sampling of the scenarios. A sufficient number of scenarios guaranteeing confidence $1-\beta$ is given by the explicit closed-form bound~\eqref{eq:sample_complexity}.
\begin{equation}
    {N_{\mathrm{scen}}} \geq \frac{2}{\varepsilon}\left(\ln\frac{1}{\beta} + n_{\mathrm{opt}}\right).
    \label{eq:sample_complexity}
\end{equation}
This bound is distribution-free: it holds for any $\mathbb{P}$, requires no structural assumptions on the uncertainty, such as Gaussianity or boundedness, and scales only logarithmically in $1/\beta$, making the confidence parameter inexpensive to tighten. The scenario approach has been successfully integrated with MPC to handle probabilistic constraints arising from stochastic disturbances~\cite{Schildbach_2014}, where each scenario corresponds to a sampled disturbance trajectory over the prediction horizon, yielding a tractable convex program with inherited feasibility guarantees.

\section{Scenario-DeePC}
\label{sec:Scenario-DeePC}
Consider the noise- and disturbance-affected system given in~\eqref{eq:system}. The objective is to control this system using the DeePC framework. While the regularization terms of DeePC in~\eqref{eq:DeePC} enhance robustness with respect to noise and disturbances in the implicit data-driven identification step, these do not explicitly account for uncertainty in the enforcement of output constraints.
In particular, measurement noise and external disturbances may lead to violations of the output constraints, even if the predicted trajectories satisfy the nominal constraints given in~\eqref{eq:DeePC_data}. To address this issue, we consider a chance-constrained formulation, where the output is required to satisfy the constraints with a prescribed probability level.
To obtain a tractable formulation, the chance constraints are approximated using the scenario approach. Specifically, $N_{\mathrm{scen}}$ scenarios are introduced to enforce the output being within the given set $\mathcal{Y}$ for each sampled realization. Furthermore, the scenario-dependent outputs $y^{(i)}$ are incorporated into the cost function to improve tracking performance under biased (non-zero-mean) uncertainties.
This results in the optimization problem in~\eqref{eq:Scenario-DeePC} that provides probabilistic guarantees on constraint satisfaction given in Section~\ref{sec:Scenario}, while remaining compatible with the DeePC framework of Section~\ref{sec:DeePC}.
\begin{subequations}\label{eq:Scenario-DeePC}
\begin{align}
\min_{g,\,u,\,y,\,\sigma_y}\quad
& \scalebox{0.8}{$\displaystyle
\frac{1}{N_{\mathrm{scen}}}\sum_{i=1}^{N_{\mathrm{scen}}}\|r - y^{(i)}\|_Q^2
+ \|u\|_R^2 + \|g\|_{\lambda_g}^2 + \|\sigma_y\|_{\lambda_y}^2
$}
\label{eq:ScDeePC_cost}\\
\text{s.t.}\quad
& \left(\begin{array}{c} U_{p} \\ Y_{p} \\ U_{f} \\ Y_{f} \end{array}\right) g
= \left(\begin{array}{c} u_{\mathrm{ini}} \\ y_{\mathrm{ini}} \\ u \\ y_{\mathrm{pred}} \end{array}\right)
+ \left(\begin{array}{c} 0 \\ \sigma_y \\ 0 \\ 0 \end{array}\right)
\label{eq:ScDeePC_data}\\
& y_k^{(i)} = y_{\mathrm{pred},k} + s_k^{(i)}
\label{eq:ScDeePC_scenario}\\
& u_k \in \mathcal{U},\quad y_k^{(i)} \in \mathcal{Y}
\label{eq:ScDeePC_constraints}\\
& \forall k \in \{0, \ldots, N-1\}\ \forall i \in \{1, \ldots, N_{\mathrm{scen}}\} \notag
\end{align}
\end{subequations}
Here $y_{\mathrm{pred}} \in \mathbb{R}^{n_y N}$ denotes the predicted output sequence of the data-driven predictor~\eqref{eq:ScDeePC_data}, and $y_{\mathrm{pred},k} \in \mathbb{R}^{n_y}$ its $k$-th step. Each scenario perturbs this prediction step-wise: $s_k^{(i)} \in \mathbb{R}^{n_y}$ is the single-step error applied at step $k$, and stacking these over the horizon yields the scenario vector $s^{(i)} := \bigl[\, s_0^{(i)}; \dots; s_{N-1}^{(i)} \,\bigr] \in \mathbb{R}^{n_y N}$, so that $y^{(i)} = y_{\mathrm{pred}} + s^{(i)} \in \mathbb{R}^{n_y N}$ is the output associated with the $i$-th scenario.
The single-step errors are drawn from a buffer of stored one-step prediction errors $\mathcal{E}_{\mathrm{buffer}} = \{w^{(1)}, \dots, w^{(N_{\mathrm{buffer}})}\}$, with each entry $w^{(j)} \in \mathbb{R}^{n_y}$. During each optimization step, each horizon step $k$ of each scenario $i$ draws independently and uniformly from this buffer, i.e.,
\begin{equation}
s_k^{(i)} \sim \mathrm{Unif}\!\left(\mathcal{E}_{\mathrm{buffer}}\right),
\quad i = 1,\dots,N_{\mathrm{scen}},\; k = 0,\dots,N-1,
\label{eq:scenario_sampling}
\end{equation}
where the sampling is performed uniformly at random.

\subsection{Data Collection for DeePC}
The data used to construct the Hankel matrices follows the standard DeePC procedure as described in Section~\ref{sec:DeePC}. In particular, an input sequence $u^d$ of ${T_d}$ discrete time steps is applied to the system in~\eqref{eq:system}, which is chosen to be persistently exciting of sufficiently high order. The corresponding output trajectory $y^d$ is obtained from the system and inherently reflects the presence of measurement noise and disturbances.
This data collection procedure reflects typical real-world settings, where measurements are inherently affected by noise and disturbances.

\subsection{Scenario Generation}
In contrast to classical robust or stochastic control approaches, no prior knowledge about the distribution of the measurement noise $v$ or the disturbance $d$ is assumed. This is in line with the data-driven philosophy of DeePC, where system properties are not modeled explicitly but inferred directly from data. Consequently, the uncertainty affecting the system is not described through an assumed modeled distribution, but constructed directly from observed data.
Therefore, we generate scenarios based on the prediction error of the DeePC controller during closed-loop operation. At each time step, DeePC computes an optimal control input sequence $u$, from which only the first control input is applied to the system, as in standard receding horizon control. Correspondingly, the data-driven predictor provides a predicted output sequence $y_{\mathrm{pred}}$ where the first output value $y_{\mathrm{pred},0}$ is extracted and compared to the measured system output $y$ obtained after applying the control input to the system at the current time step, resulting in the error $w^{(j)}$ given in~\eqref{eq:prediction_error}.
The discrepancy between prediction and measurement is then interpreted as a realization of the uncertainty affecting the system at the current time step.
\begin{equation}
w^{(j)} = y - y_{\mathrm{pred},0}.
\label{eq:prediction_error}
\end{equation}
In this sense, the scenario set directly represents the empirical distribution of the uncertainty acting on the closed-loop system.
Importantly, this prediction error $w^{(j)}$ represents the combined effect of external disturbances $d$, measurement noise $v$, and residual data-driven modeling inaccuracies arising from the noisy and disturbance-affected data in the Hankel matrix.
To make this composition explicit, recall that every length-$L$ column of the Hankel matrices is a system trajectory, so the output Hankel matrix admits the behavioral decomposition
\begin{equation} 
\mathcal{H}_L(y^d) \;=\; \Gamma_L\, X_0 \;+\; \mathcal{T}_L\, \mathcal{H}_L(u^d) \;+\; V, 
\label{eq:data_equation} 
\end{equation}
where $\Gamma_L = \bigl[\, C^\top \;\; (CA)^\top \;\; \cdots \;\; (CA^{L-1})^\top \,\bigr]^\top$ is the extended observability matrix, $\mathcal{T}_L$ is the lower-block-triangular Toeplitz matrix of Markov parameters $\{D,\, CB,\, CAB,\, \ldots\}$, the columns of $X_0$ collect the per-window initial states, and $V$ is the Hankel matrix of the combined disturbance and noise sequence $\{d_k + v_k\}$. Its future block specializes to 
\begin{equation}
    Y_f = \Gamma_N X_f + \mathcal{T}_N U_f + V_f, 
\end{equation}
so the data-driven prediction reads
\begin{equation} 
y_\text{pred} \;=\; Y_f\, g \;=\; \Gamma_N\, \xi_g \;+\; \mathcal{T}_N\, u \;+\; V_f\, g, 
\label{eq:pred_decomp} 
\end{equation}
where $\xi_g = X_f\, g$ is the effective initial state selected by $g$ through the past-block constraints $U_p g = u_\text{ini}$ and $Y_p g = y_\text{ini} + \sigma_y$, and $X_f$ collects the states at the start of the future window of each Hankel column. Under noise-free data and persistency of excitation, Willems' fundamental lemma yields $\xi_g = x$ and $V_f = 0$, so~\eqref{eq:pred_decomp} reduces to the exact predictor, and the prediction error vanishes.

In the noise-affected setting of~\eqref{eq:system}, neither condition holds. Since the realized output satisfies 
\begin{equation}
    y = \Gamma_N x + \mathcal{T}_N u + (d + v),
\end{equation} 
and the planned input $u$ cancels between measurement and prediction, the prediction error~\eqref{eq:prediction_error} decomposes as
\begin{equation}
\begin{aligned}
y - y_{\mathrm{pred}} &= \underbrace{(d + v)}_{\text{disturbance}\,+\,\text{noise}} - \underbrace{\Gamma_N\bigl(\xi_g - x\bigr) - V_f\, g}_{\text{data-driven identification residual}}, \\[4pt]
w^{(j)} &= \bigl[\, y - y_{\mathrm{pred}} \,\bigr]_{k=0} \in \mathbb{R}^{n_y},
\end{aligned}
\label{eq:error_decomp}
\end{equation}
where $[\cdot]_{k=0}$ denotes the first $n_y$-block (the one-step-ahead entry) of the stacked prediction error. The first term is the instantaneous disturbance and measurement noise. The second is a residual produced by the noise embedded in the Hankel data itself. The noisy past block 
\begin{equation}
    Y_p = \Gamma_{T_\text{ini}} X_0 + \mathcal{T}_{T_\text{ini}} \mathcal{H}_{T_\text{ini}}(u^d) + V_p
\end{equation}
biases the reconstructed state $\xi_g$, while $V_f g$ propagates the future-block noise directly into the prediction. Crucially, the slack $\sigma_y$ enters~\eqref{eq:error_decomp} only through $\xi_g$. Therefore, it relaxes the past-output matching but cannot de-noise the Hankel columns, and the regularizer $\|g\|_{\lambda_g}^2$ shrinks but does not eliminate this residual~\cite{Dorfler_2022}. A prediction error therefore persists on top of $(d+v)$ regardless of the slacking and regularization, and for nonlinear systems an additional model-mismatch term enters through the LTI parametrization. This composite residual, and not the measurement noise alone, is precisely the uncertainty the controller must be robustified against, which motivates collecting the scenarios directly from the observed prediction errors~\eqref{eq:prediction_error} rather than from an assumed model of the sensor noise.
While regularization in DeePC mitigates these effects strongly in the prediction step, it does not eliminate them entirely. By directly using the observed prediction error as a scenario, the proposed approach captures the true uncertainty acting on the system in a fully data-driven manner, without relying on simplifying assumptions or conservative approximations.

The collection of such scenarios follows a procedure analogous to the data collection in standard DeePC. Instead of operating the system in open loop as in standard DeePC, the closed-loop system consisting of the DeePC controller and the physical system is run along a predefined reference trajectory. During this phase, both the predicted outputs $y_{\mathrm{pred},0}$ and the corresponding measured outputs $y$ are recorded. The resulting prediction errors, computed according to~\eqref{eq:prediction_error}, are stored in the scenario buffer $\mathcal{E}_{\mathrm{buffer}}$. Consecutive prediction errors are temporally correlated through the closed-loop dynamics, whereas the probabilistic guarantee requires independent samples. We therefore retain only every $(M+1)$-th error, where $M$ is the correlation horizon, leaving consecutive buffer entries more than $M$ closed-loop steps apart and hence independent. This $M$-dependence, and the guarantee it enables, are made precise in Section~\ref{sec:Theo_Results}.
During operation, the Scenario-DeePC controller samples $N_{\mathrm{scen}}$ scenarios from this buffer as described previously. These randomly sampled prediction errors are then incorporated into the constraint formulation, enabling a scenario-based approximation of the underlying chance constraints. \\
In the perturbation~\eqref{eq:ScDeePC_scenario}, each prediction step $k$ is perturbed by its own sampled one-step scenario $s_k^{(i)}$, so that scenario $i$ is the horizon-length sequence $(s_0^{(i)}, \dots, s_{N-1}^{(i)})$ assembled from independent draws of~\eqref{eq:scenario_sampling}. Drawing the per-step blocks independently from the empirical one-step marginal $\hat{\mathbb{P}}_w := \mathrm{Unif}(\mathcal{E}_{\mathrm{buffer}})$ makes the $N$ horizon steps of $s^{(i)}$ mutually independent, each distributed as $\hat{\mathbb{P}}_w$, so the scenario program, and hence the guarantee of Theorem~\ref{thm:scenario_guarantee}, robustifies against this product measure over the horizon rather than against the joint distribution of consecutive prediction errors. Because the controller re-optimizes at every step, the closed-loop collection observes only the one-step-ahead error~\eqref{eq:prediction_error}. A multi-step error trajectory under a frozen input plan is never realized and hence cannot be collected. The per-step marginal is thus the only distributional information available, so assembling each scenario from independent per-step draws is the natural choice.
This formulation ensures that the predicted outputs remain within the admissible set $\mathcal{Y}$ for all considered realizations of the uncertainty, thereby providing a tractable approximation of the underlying chance constraints.

\subsection{Adaptive Scenario-DeePC}
Since the scenarios are constructed from prediction errors observed during closed-loop operation in a receding horizon fashion, it is natural to update the scenario set online while the controller is running. Therefore, a replay buffer $\mathcal{E}_{\mathrm{buffer}}$ of size $N_{\mathrm{buffer}}$ is maintained, which stores the most recent prediction error samples.
Every $(M+1)$-th closed-loop step, a new prediction error is computed according to~\eqref{eq:prediction_error} and added to the buffer, while the oldest stored error is discarded. This results in a regularly updated set of scenario realizations that reflects the current operating conditions of the system.
Such an online scenario generation mechanism enables the controller to stay robust to changes in the system and its environment, including variations in measurement noise, the occurrence of new disturbances, or data-driven model mismatch due to system degradation, and can improve the practical robustness of constraint handling. In particular, increasing prediction errors are directly captured in the scenario set, leading to an automatic adjustment of the constraint tightening through the scenario-based formulation, while decreasing prediction errors result in less conservative constraint tightening.
\begin{remark}
A non-uniform sampling of the replay buffer, giving higher weight to recent scenarios, could improve adaptation to fast-changing disturbances. However, such adaptive sampling strategies are beyond the scope of this work, and uniform sampling is used throughout.  
\end{remark}

\section{Theoretical Results}
\label{sec:Theo_Results}
We first establish that Scenario-DeePC is a well-posed convex program and quantify its decision dimension, which underpins the probabilistic guarantee derived in the remainder of this section.
Throughout this section, we make the standard standing assumptions of data-driven predictive control.
\begin{assumption}[Data-generating system and excitation]
\label{ass:standing}
The offline data $(u^d, y^d)$ is generated by a minimal (controllable and observable) LTI system~\eqref{eq:system} of order $n_x$, where $n_x$ may be replaced by a known upper bound. The initial horizon satisfies $T_{\mathrm{ini}} \ge n_x$, and the input $u^d$ is persistently exciting of order $T_{\mathrm{ini}} + N + n_x$, so that the rank condition~\eqref{eq:pe_rank} holds.
\end{assumption}
\subsection{{Well-Posedness and Decision Dimension}}
\begin{assumption}[Convex constraint sets]
\label{ass:convex_sets}
The input and output constraint sets $\mathcal{U} \subseteq \mathbb{R}^{n_u}$ and $\mathcal{Y} \subseteq \mathbb{R}^{n_y}$ are convex and closed.
\end{assumption}

\begin{lemma}[Convexity and solvability]
\label{lem:convexity}
Let Assumption~\ref{ass:convex_sets} hold with $Q \succ 0$, $R \succ 0$ and $\lambda_g, \lambda_y > 0$. Then the Scenario-DeePC problem~\eqref{eq:Scenario-DeePC} is a convex quadratic program, and whenever it is feasible it attains a global minimizer.
\end{lemma}

\begin{proof}
Collect the decision variables in $z := (g, u, y_{\mathrm{pred}}, \sigma_y)$.

\emph{Objective:} With $y^{(i)} = y_{\mathrm{pred}} + s^{(i)}$, the scenarios enter the cost~\eqref{eq:ScDeePC_cost} only as constant offsets. Each term $\|r - y^{(i)}\|_Q^2$, $\|u\|_R^2$, $\|g\|_{\lambda_g}^2$ and $\|\sigma_y\|_{\lambda_y}^2$ is then a convex quadratic in $z$. The first is convex by $Q \succ 0$, composed with the affine map $y_{\mathrm{pred}} \mapsto y_{\mathrm{pred}} + s^{(i)} - r$, the second by $R \succ 0$, and the last two by $\lambda_g, \lambda_y > 0$. Their sum~\eqref{eq:ScDeePC_cost} is therefore convex. Since the four terms are strictly convex in $y_{\mathrm{pred}}$, $u$, $g$ and $\sigma_y$ respectively, which are disjoint components of $z$, the objective is strictly convex. Its minimizer, when attained, is unique.

\emph{Feasible set:} The equality constraints~\eqref{eq:ScDeePC_data}, are affine in $z$ and define an affine subspace. By~\eqref{eq:ScDeePC_constraints}, $u_k \in \mathcal{U}$ and $y_{\mathrm{pred},k} \in \mathcal{Y} - s^{(i)}_k$ for all $i,k$. Under Assumption~\ref{ass:convex_sets}, these are intersections of convex closed sets, so the feasible set, their intersection with the affine subspace, is thus convex and closed.

\emph{Solvability:} Eliminating $u, y_{\mathrm{pred}}, \sigma_y$ via~\eqref{eq:ScDeePC_data} leaves a reduced cost in $g$ that is coercive, as $\|g\|_{\lambda_g}^2 \to \infty$ for $\|g\| \to \infty$ with $\lambda_g > 0$. A coercive, strictly convex quadratic over a nonempty closed convex set attains a unique global minimizer.
\end{proof}

\begin{lemma}[Support-dimension bound]
\label{lem:support_dim}
The number of support constraints of the Scenario-DeePC problem~\eqref{eq:Scenario-DeePC} is at most
\begin{equation}
n_{\mathrm{opt}} = T_d - T_{\mathrm{ini}} - N + 1,
\label{eq:nopt}
\end{equation}
the dimension of its only free decision variable $g$.
\end{lemma}

\begin{proof}
The equality constraints~\eqref{eq:ScDeePC_data} express every remaining variable as an affine image of $g$,
\begin{equation}
u = U_f g, \qquad y_{\mathrm{pred}} = Y_f g, \qquad \sigma_y = Y_p g - y_{\mathrm{ini}},
\end{equation}
while $g$ itself is confined to the affine set $\{\, g : U_p g = u_{\mathrm{ini}} \,\}$. In particular, the slack $\sigma_y$ is not an independent variable: the equality $Y_p g = y_{\mathrm{ini}} + \sigma_y$ pins it to the past-output residual, so its penalty
\begin{equation}
\|\sigma_y\|_{\lambda_y}^2 = \|Y_p g - y_{\mathrm{ini}}\|_{\lambda_y}^2
\end{equation}
is a regularizer on $g$ rather than a degree of freedom of its own, and the same holds for $u$ and $y_{\mathrm{pred}}$. Hence $g \in \mathbb{R}^{\,T_d - T_{\mathrm{ini}} - N + 1}$ is the only free decision variable of~\eqref{eq:Scenario-DeePC}, and the certain constraint $U_p g = u_{\mathrm{ini}}$ can only lower the count further. Since the uncertain constraints~\eqref{eq:ScDeePC_scenario}--\eqref{eq:ScDeePC_constraints} act on the decision solely through $y_{\mathrm{pred}} = Y_f g$, the scenario theory~\cite{Campi_2008} bounds the number of support constraints of the convex program by its number of decision variables, here $\dim(g) = n_{\mathrm{opt}}$.
\end{proof}
By Lemmas~\ref{lem:convexity} and~\ref{lem:support_dim}, the Scenario-DeePC problem~\eqref{eq:Scenario-DeePC} is a convex program whose support dimension is at most $n_{\mathrm{opt}}$, so the scenario approach~\cite{Campi_2008} applies: it bounds, distribution-free, the probability that the optimal solution violates the output constraint on a realization drawn from the empirical error buffer. Its classical form presumes i.i.d.\ scenarios from a fixed distribution. In Scenario-DeePC, the scenarios are instead the empirical prediction errors~\eqref{eq:prediction_error} gathered during closed-loop operation, which are temporally correlated through the closed-loop dynamics. We therefore impose an $M$-dependence assumption on the error process and recover independent scenarios by retaining only every $(M+1)$-th error, after which the classical scenario bound applies directly.

\subsection{{Probabilistic Constraint Satisfaction}}
\begin{assumption}[$M$-dependence of the prediction-error process]
\label{ass:mixing}
After an initial transient, the prediction-error process $\{w^{(j)}\}$ defined in~\eqref{eq:prediction_error} is stationary and $M$-dependent~\cite{Bosq_1998} for some $M \in \mathbb{N}$: for every $\ell$, the past $\{w^{(j)} : j \leq \ell\}$ and the future $\{w^{(j)} : j \geq \ell + M + 1\}$ are independent. Equivalently, errors separated by more than $M$ closed-loop steps are mutually independent.
\end{assumption}

\begin{theorem}[Probabilistic constraint satisfaction]
\label{thm:scenario_guarantee}
Let Assumptions~\ref{ass:convex_sets} and~\ref{ass:mixing} hold, and let the buffer be populated with every $(M+1)$-th closed-loop prediction error. For $\varepsilon, \beta \in (0,1)$, the optimal Scenario-DeePC solution satisfies
\begin{equation}
\begin{aligned}
& \mathbb{P}^{N_{\mathrm{scen}}}\!\left\{ V > \varepsilon \right\} \;\leq\; \beta, \\[4pt]
& V := \mathbb{P}_{s}\!\left\{ \exists\, k \in \{0,\dots,N-1\} : y_{\mathrm{pred},k}^\star + s_k \notin \mathcal{Y} \right\},
\end{aligned}
\end{equation}
whenever the number of scenarios obeys the standard bound~\eqref{eq:sample_complexity} with $n_{\mathrm{opt}}$ from~\eqref{eq:nopt}, where $V$ is the probability, under the empirical buffer distribution, that a sampled prediction error drives the predicted output outside $\mathcal{Y}$. These $N_{\mathrm{scen}}$ scenarios are drawn from a buffer of $N_{\mathrm{buffer}} \geq N_{\mathrm{scen}}$ independent errors, and filling the buffer at stride $M+1$ requires
\begin{equation}
N_{\mathrm{cl}} \;\geq\; (M+1)\,N_{\mathrm{buffer}}
\label{eq:sample_complexity_mixing}
\end{equation}
consecutive closed-loop steps, where $N_{\mathrm{cl}}$ is the number of closed-loop steps recorded.
\end{theorem}

\begin{proof}
By Assumption~\ref{ass:mixing}, prediction errors separated by more than $M$ steps are independent. Retaining every $(M+1)$-th error therefore yields buffer entries that are mutually independent, and stationarity makes them identically distributed. The scenarios drawn as in~\eqref{eq:scenario_sampling} are thus an i.i.d.\ sample of the empirical buffer distribution. By Lemma~\ref{lem:convexity} the program~\eqref{eq:Scenario-DeePC} is convex and by Lemma~\ref{lem:support_dim} its support dimension is at most $n_{\mathrm{opt}}$, so the classical scenario approach~\cite{Campi_2008} applies and yields $\mathbb{P}^{N_{\mathrm{scen}}}\{V > \varepsilon\} \leq \beta$ under~\eqref{eq:sample_complexity}. Since one error is retained every $M+1$ steps, filling the buffer with $N_{\mathrm{buffer}}$ independent errors consumes $(M+1)\,N_{\mathrm{buffer}}$ consecutive closed-loop steps. 
\end{proof}

Theorem~\ref{thm:scenario_guarantee} certifies constraint satisfaction with respect to the empirical distribution of the collected prediction-error buffer, without any assumption on the underlying error distribution $\mathbb{P}$. This finite-sample, distribution-free certificate is the intended guarantee of the data-driven scenario formulation. It holds for the scenarios actually observed, and the buffer represents the operating conditions more faithfully as more prediction errors are collected.
The bound~\eqref{eq:sample_complexity} scales linearly in $n_{\mathrm{opt}}$ and logarithmically in $1/\beta$, so the confidence level is inexpensive to tighten. The $M$-dependence enters only through the $(M+1)$ factor in the data-collection horizon~\eqref{eq:sample_complexity_mixing}, not in the guarantee itself. The $M$-dependent setting is also treated in~\cite{Campi_2009_2}. In DeePC, $n_{\mathrm{opt}}$ is governed by the number of Hankel columns and is therefore large, rendering the bound considerably more conservative than in model-based scenario MPC~\cite{Campi_2009}. The numerical results in Section~\ref{sec:Num_Results} show that substantially fewer scenarios already yield satisfactory constraint satisfaction.

\begin{remark}[Mixing in practice]
\label{rem:mixing}
Assumption~\ref{ass:mixing} is reasonable: for a well-performing predictor, the prediction error is dominated by the time-independent measurement noise, while the stable closed loop forgets past noise geometrically, so its correlation decays with the lag, and treating the errors as approximately independent is standard in the scenario literature~\cite{Campi_2021}. The horizon $M$ follows from the empirical autocorrelation of the buffered errors (Section~\ref{sec:Num_Results}) or the closed-loop settling time, and since~\eqref{eq:sample_complexity_mixing} is monotone in $M$, overestimating it stays valid. In the adaptive variant stationarity fails, so Theorem~\ref{thm:scenario_guarantee} does not apply. Formal guarantees are left to future work, though the empirical results still show robust constraint satisfaction.
\end{remark}

\subsection{{Recursive Feasibility}}
The Scenario-DeePC problem in~\eqref{eq:Scenario-DeePC} enforces the output constraints $y_k^{(i)} \in \mathcal{Y}$ for all scenarios $i$ and all prediction steps $k$. When the prediction errors are large, or when the system is driven close to the constraint boundary, no feasible solution to~\eqref{eq:Scenario-DeePC} may exist, and the receding-horizon implementation can break down. To recover \textit{feasibility by construction}, we introduce a non-negative slack variable $h \in \mathbb{R}^{n_y}_{\geq 0}$ that relaxes the output constraint set component-wise:
\begin{equation}
    \mathcal{Y}(h)
    \;:=\;
    \bigl\{
        y \in \mathbb{R}^{n_y}
        \;\big|\;
        \underline{y} - h \;\leq\; y \;\leq\; \bar{y} + h
    \bigr\},
    \label{eq:relaxed_set}
\end{equation}
where $\underline{y}$ and $\bar{y}$ are the original lower and upper output bounds defining $\mathcal{Y}$.
The relaxed Scenario-DeePC problem is then
\begin{subequations}\label{eq:ScDeePC_relaxed}
\begin{align}
\min_{g,\,u,\, y,\,\sigma_y,\,h}\quad
& \begin{aligned}[t]
    & \frac{1}{N_\text{scen}} \sum_{i=1}^{N_\text{scen}} \|r - y^{(i)}\|_Q^2
      + \|u\|_R^2 + \|g\|_{\lambda_g}^2 \\
    & \quad + \|\sigma_y\|_{\lambda_y}^2 + \mu \|h\|_1
  \end{aligned}
\label{eq:ScDeePC_relaxed_cost}\\
\text{s.t.}\quad
& \begin{bmatrix} U_p \\ Y_p \\ U_f \\ Y_f \end{bmatrix} g
= \begin{bmatrix} u_\text{ini} \\ y_\text{ini} \\ u \\ y_\text{pred} \end{bmatrix}
+ \begin{bmatrix} 0 \\ \sigma_y \\ 0 \\ 0 \end{bmatrix}
\label{eq:ScDeePC_relaxed_data}\\
& y_k^{(i)} = y_{\text{pred},k} + s_k^{(i)}
\label{eq:ScDeePC_relaxed_scenario}\\
& h \geq 0 \quad \\
& u_k \in \mathcal{U} \quad y_k^{(i)} \in \mathcal{Y}(h)
\label{eq:ScDeePC_relaxed_constraints}\\
& \forall\, k \in \{0,\ldots,N-1\}\;
\forall\, i \in \{1,\ldots,N_\text{scen}\}, \notag
\end{align}
\end{subequations}
where $\mu > 0$ is a penalty weight. The $\ell_1$ form makes it an exact penalty, so for $\mu$ sufficiently large the slack vanishes whenever the unrelaxed problem is feasible (Proposition~\ref{prop:recursive_feasibility}(iii)).
\begin{proposition}[Recursive feasibility of relaxed Scenario-DeePC]
\label{prop:recursive_feasibility}
Consider the relaxed Scenario-DeePC problem~\eqref{eq:ScDeePC_relaxed} with $\mu > 0$, implemented in a receding-horizon fashion. The following statements hold:
\begin{enumerate}[(i)]
    \item \emph{Feasibility by construction:} Problem~\eqref{eq:ScDeePC_relaxed} is feasible at every time step $t$, for any realization of the scenario set $\{s^{(i)}\}_{i=1}^{N_\text{scen}}$ and any initial condition $(u_\text{ini}, y_\text{ini})$.

    \item \emph{Recursive feasibility:} The receding-horizon implementation of relaxed Scenario-DeePC is recursively feasible: if~\eqref{eq:ScDeePC_relaxed} is feasible at time step $t$, it is feasible at every subsequent time step $t' > t$.

    \item \emph{Reduction to original Scenario-DeePC:} For $\mu$ sufficiently large, whenever the original Scenario-DeePC problem~\eqref{eq:Scenario-DeePC} is feasible (i.e., there exists a solution satisfying $y_k^{(i)} \in \mathcal{Y}$ for all $i, k$), the optimal slack satisfies $h^\star = 0$, and problem~\eqref{eq:ScDeePC_relaxed} coincides with~\eqref{eq:Scenario-DeePC}.
\end{enumerate}
\end{proposition}

\begin{proof}
We decompose the feasibility of~\eqref{eq:ScDeePC_relaxed} into two subproblems: the data-equality constraint~\eqref{eq:ScDeePC_relaxed_data} inherited from standard DeePC, and the scenario output constraints~\eqref{eq:ScDeePC_relaxed_scenario}~-~\eqref{eq:ScDeePC_relaxed_constraints} introduced by the scenario approach.

\emph{(i) Feasibility by construction:}

\emph{Subproblem~1 (data equality constraints):} For any fixed $(u_\text{ini}, y_\text{ini})$ and any decision input $u$, the persistency-of-excitation condition~\eqref{eq:pe_rank} guarantees that the input Hankel block $\bigl[\begin{smallmatrix} U_p U_f \end{smallmatrix}\bigr]^\top = \mathcal{H}_L(u^d)$ has full row rank, so the input equalities $U_p g = u_\text{ini}$ and $U_f g = u$ admit a solution $g$. The slack $\sigma_y$ renders the past-output equality $Y_p g = y_\text{ini} + \sigma_y$ solvable for any such $g$, and $y_\text{pred} = Y_f g$ is then determined. Hence, the data-equality block is feasible at every time step and for every scenario realization, a property inherited directly from standard DeePC~\cite{Coulson_2019, Willems_2004}. We further note that, since the output data is noise-affected, the stacked Hankel matrix $[U_p;\,Y_p;\,U_f;\,Y_f]$ is generically of full row rank, so the data-equality block is solvable for any right-hand side even without $\sigma_y$. Measurement noise, which generally degrades prediction accuracy, therefore does not obstruct the feasibility of the data-equality subproblem.

\emph{Subproblem~2 (scenario output constraints):} Given any such $g$, the prediction $y_\text{pred}$ and the scenario outputs $y^{(i)} = y_\text{pred} + s^{(i)}$ are fixed. Choosing the slack component-wise as
\begin{equation}
    h \;=\; \max_{i,\,k}\;
    \Bigl(
        \bigl(y_k^{(i)} - \bar{y}\bigr)^+
        \;\vee\;
        \bigl(\underline{y} - y_k^{(i)}\bigr)^+
    \Bigr)
    \;\geq\; 0,    
\end{equation}
where $(\cdot)^+$ and $\vee$ denote the component-wise positive part and maximum, renders every scenario constraint $y_k^{(i)} \in \mathcal{Y}(h)$ satisfied. Since $h \geq 0$ and $\mu > 0$, the objective is finite. Combining both subproblems, a feasible point of~\eqref{eq:ScDeePC_relaxed} exists for any initial condition and any scenario set, establishing~(i).

\emph{(ii) Recursive feasibility:} By~(i), problem~\eqref{eq:ScDeePC_relaxed} is feasible for arbitrary $(u_\text{ini}, y_\text{ini})$ and arbitrary scenario sets. Feasibility is therefore not merely propagated from one step to the next but holds unconditionally. Recursive feasibility follows as an immediate corollary, since the initial condition formed at every closed-loop step and the scenario set drawn from $\mathcal{E}_\text{buffer}$ are admissible instances of~(i).

\emph{(iii) Reduction to original Scenario-DeePC:}
When~\eqref{eq:Scenario-DeePC} is feasible, its minimizer is feasible for~\eqref{eq:ScDeePC_relaxed} with $h = 0$. It remains to show the penalty forces $h^\star = 0$. Collect the decision variables in $z = (g, u, y, \sigma_y, h)$ and recall~\eqref{eq:relaxed_set}. Dualizing the relaxed output constraints of~\eqref{eq:ScDeePC_relaxed} component-wise, with multipliers $\bar\lambda^{(i)}_{k,j} \ge 0$ for the upper bound $y^{(i)}_{k,j} \le \bar{y}_j + h_j$ and $\underline\lambda^{(i)}_{k,j} \ge 0$ for the lower bound $y^{(i)}_{k,j} \ge \underline{y}_j - h_j$, gives the partial Lagrangian
\begin{equation}
\scalebox{0.85}{$\displaystyle
\mathcal{L} \!=\! J(z) \!+\! \sum_{i=1}^{N_{\mathrm{scen}}} \sum_{k=0}^{N-1} \sum_{j=1}^{n_y} \Bigl[ \bar\lambda^{(i)}_{k,j}\bigl( y^{(i)}_{k,j} \!-\! \bar{y}_j \!-\! h_j \bigr) \!+\! \underline\lambda^{(i)}_{k,j}\bigl( \underline{y}_j \!-\! h_j \!-\! y^{(i)}_{k,j} \bigr) \Bigr],
$}
\label{eq:lagrangian}
\end{equation}
where $J(z)$ is the cost~\eqref{eq:ScDeePC_relaxed_cost}, containing the penalty $\mu\sum_j h_j$. The single slack $h_j$ enters both bounds, so collecting its terms in~\eqref{eq:lagrangian} yields the contribution $(\mu - \Lambda_j)\,h_j$, with the aggregate multiplier
\begin{equation}
\Lambda_j \;:=\; \sum_{i=1}^{N_{\mathrm{scen}}} \sum_{k=0}^{N-1} \bigl( \bar\lambda^{(i)}_{k,j} + \underline\lambda^{(i)}_{k,j} \bigr) \;\ge\; 0 .
\label{eq:agg_multiplier}
\end{equation}
Minimized over $h_j \ge 0$, this contribution attains $h_j^\star = 0$ whenever $\mu \ge \Lambda_j$, and the $\ell_1$ penalty is thus exact~\cite{Kerrigan_2000}. Hence for every $\mu > \bar\mu := \max_j \Lambda_j$, the optimal slack vanishes, $h^\star = 0$, and~\eqref{eq:ScDeePC_relaxed} coincides with~\eqref{eq:Scenario-DeePC}. 
\end{proof}
The optimal slack $h^\star$ at each time step quantifies the minimum relaxation of the output constraint set required for feasibility given the current scenario realization. When $h^\star = 0$, the original output constraints are satisfied for all scenarios and the relaxed formulation is equivalent to~\eqref{eq:Scenario-DeePC}. When $h^\star > 0$, the constraint set is automatically widened by the smallest amount necessary, trading strict constraint satisfaction for guaranteed solvability. With the $\ell_1$ penalty, this trade-off is exact in $\mu$: once $\mu$ exceeds a finite threshold, $h^\star = 0$ whenever the unrelaxed problem is feasible, so the relaxation activates only when strictly necessary, while smaller $\mu$ permits earlier relaxation to reduce conservatism in the presence of large prediction errors. This behavior is analogous to the slack-variable relaxation employed in scenario-based stochastic MPC~\cite{Deori_2017}, adapted here to the data-driven DeePC setting.

Proposition~\ref{prop:recursive_feasibility} establishes feasibility by construction: feasibility of the data-equality block is inherited from standard DeePC under persistency of excitation, while the slack $h$ absorbs any infeasibility introduced by the scenario tightening of the output set. This differs from classical recursive feasibility results in tube-based or terminal-set MPC~\cite{Mayne_2000}, where feasibility of the original (unrelaxed) constraints is propagated forward using a terminal invariant set. We deliberately do not employ terminal ingredients. Although data-based terminal-set constructions have been proposed~\cite{Berberich_2020}, they require additional design (e.g., data-based LMIs) together with a bound on the noise, which would conflict with the distribution- and bound-free philosophy underlying Scenario-DeePC.

\begin{corollary}[Retention of the probabilistic guarantee]
\label{cor:retained_guarantee}
Whenever the unrelaxed problem~\eqref{eq:Scenario-DeePC} is feasible (equivalently, $h^\star = 0$), the relaxed solution coincides with the scenario solution and the probabilistic guarantee of Theorem~\ref{thm:scenario_guarantee} for the original output set $\mathcal{Y}$ is retained with complexity $n_{\mathrm{opt}}$. When $h^\star > 0$, the same guarantee applies to the widened set $\mathcal{Y}(h^\star)$ in place of $\mathcal{Y}$, with complexity at most $n_{\mathrm{opt}} + n_y$. Therefore, the slack never invalidates the probabilistic certificate. It only relocates it from $\mathcal{Y}$ to the minimally enlarged set $\mathcal{Y}(h^\star)$ required for feasibility at the current step.
\end{corollary}
\begin{proof}
For $h^\star = 0$, Proposition~\ref{prop:recursive_feasibility}(iii) makes the relaxed problem~\eqref{eq:ScDeePC_relaxed} coincide with~\eqref{eq:Scenario-DeePC}, so Theorem~\ref{thm:scenario_guarantee} applies verbatim with support dimension $n_{\mathrm{opt}}$. For $h^\star > 0$, the relaxed program adds the slack $h \in \mathbb{R}^{n_y}_{\geq 0}$ to the free decision variables, so by the support-dimension argument of Lemma~\ref{lem:support_dim} its support dimension is at most $n_{\mathrm{opt}} + n_y$. Theorem~\ref{thm:scenario_guarantee} then applies with $n_{\mathrm{opt}}$ replaced by $n_{\mathrm{opt}} + n_y$, certifying the same probabilistic guarantee for the relaxed constraint, i.e.\ $\mathbb{P}_{s}\!\left\{ \exists\, k \in \{0,\dots,N-1\} : y_{\mathrm{pred},k}^\star + s_k \notin \mathcal{Y}(h^\star) \right\} \leq \varepsilon$ with confidence $1-\beta$.
\end{proof}

\subsection{{Consistency with Standard DeePC}}
\begin{theorem}[Equivalence of Scenario-DeePC and DeePC under LTI conditions]
\label{thm:zero_error_equivalence}
Consider the LTI system~\eqref{eq:system} with $d_k = 0$ and $v_k = 0$ for all $k$, and let the data underlying the Hankel matrices be persistently exciting of sufficient order in the sense of~\eqref{eq:pe_rank}. Then the prediction error~\eqref{eq:prediction_error} vanishes, $w^{(j)} = 0$ for all collected samples $j \in \{1, \dots, N_{\mathrm{buffer}}\}$, the scenario buffer is $\mathcal{E}_{\mathrm{buffer}} = \{0\}$, and the Scenario-DeePC problem~\eqref{eq:Scenario-DeePC} reduces exactly to the standard DeePC problem~\eqref{eq:DeePC}.
\end{theorem}
\begin{proof}
By the prediction-error decomposition~\eqref{eq:error_decomp}, the full prediction error is $y - y_{\mathrm{pred}} = (d+v) - \Gamma_N(\xi_g - x) - V_f\, g$. Under $d_k = v_k = 0$ the disturbance-and-noise term vanishes, $d+v = 0$, and the data Hankel carries no noise, $V_f = 0$. Persistency of excitation~\eqref{eq:pe_rank} together with Willems' fundamental lemma~\cite{Willems_2004} makes the data-driven predictor reconstruct the true state, $\xi_g = x$. Hence $y - y_{\mathrm{pred}} = 0$, and in particular its first block
\begin{equation}
\label{eq:proof_zero_scenario}
w^{(j)} = \bigl[\, y - y_{\mathrm{pred}} \,\bigr]_{k=0} = 0
\end{equation}
for every collected sample $j$, and the scenario buffer is $\mathcal{E}_{\mathrm{buffer}} = \{w^{(1)}, \ldots, w^{(N_{\mathrm{buffer}})}\} = \{0, \ldots, 0\}$.
Any scenario drawn from this buffer according to~\eqref{eq:scenario_sampling} therefore satisfies $s_k^{(i)} = 0$ for all $i \in \{1, \dots, N_{\mathrm{scen}}\}$ and $k \in \{0, \dots, N-1\}$, so the stacked scenarios vanish, $s^{(i)} = 0$.
In~\eqref{eq:Scenario-DeePC}, the scenario-perturbed outputs reduce to
\begin{equation}
\label{eq:proof_scenario_output_collapse}
y^{(i)} = y_{\mathrm{pred}} + s^{(i)} = y_{\mathrm{pred}},
\quad \forall i \in \{1, \dots, N_{\mathrm{scen}}\}.
\end{equation}
The averaged scenario tracking cost therefore becomes
\begin{equation}
\label{eq:proof_cost_collapse}
\frac{1}{N_{\mathrm{scen}}} \sum_{i=1}^{N_{\mathrm{scen}}}
\|r - y^{(i)}\|_Q^2
= \|r - y_{\mathrm{pred}}\|_Q^2,
\end{equation}
which coincides with the tracking cost $\|r - y\|_Q^2$ in the standard DeePC problem~\eqref{eq:DeePC_cost} under the identification $y = y_{\mathrm{pred}}$.
The scenario output constraints in~\eqref{eq:ScDeePC_constraints} require $y^{(i)}_k \in \mathcal{Y}$ for all $i$ and $k$. Since $y^{(i)}_k = y_{\mathrm{pred},k}$ for all $i$ by \eqref{eq:proof_scenario_output_collapse}, the $N_{\mathrm{scen}}$ identical constraints reduce to the single constraint
\begin{equation}
\label{eq:proof_constraint_collapse}
y_{\mathrm{pred},k} \in \mathcal{Y}, \quad k = 0, \dots, N-1,
\end{equation}
which is precisely the output constraint $y_k \in \mathcal{Y}$~\eqref{eq:DeePC_constraints} of standard DeePC.
The remaining terms are identical in both formulations. Hence, problem~\eqref{eq:Scenario-DeePC} reduces exactly to problem~\eqref{eq:DeePC}.
\end{proof}

\begin{assumption}[Zero-mean prediction error process]
\label{ass:zero_mean_process}
The prediction error process $\{w^{(j)}\}$ generating the scenario buffer $\mathcal{E}_{\mathrm{buffer}}$, defined in~\eqref{eq:prediction_error} and capturing the combined effect of external disturbances, measurement noise, and residual data-driven modeling inaccuracies, has zero mean in the sense that, as $N_{\mathrm{buffer}} \to \infty$,
\begin{equation}
\label{eq:zero_mean_process}
\bar w \;:=\; \lim_{N_{\mathrm{buffer}} \to \infty}
\frac{1}{N_{\mathrm{buffer}}} \sum_{j=1}^{N_{\mathrm{buffer}}} w^{(j)}
\;=\; 0.
\end{equation}
\end{assumption}

\begin{assumption}[Interior operation]
\label{ass:interior}
The system operates sufficiently far from the output constraint boundaries that the scenario output constraints $y^{(i)}_k = y_{\mathrm{pred},k} + s^{(i)}_k \in \mathcal{Y}$ are inactive for every sampled scenario sequence $s^{(i)}$ and every step $k$ of the prediction horizon.
\end{assumption}

\begin{theorem}[Asymptotic equivalence under zero-mean scenarios]
\label{thm:zero_mean_equivalence}
Let Assumptions~\ref{ass:zero_mean_process} and~\ref{ass:interior} hold. Then, in the iterated limit $N_{\mathrm{scen}} \to \infty$ followed by $N_{\mathrm{buffer}} \to \infty$, the Scenario-DeePC cost function converges to the standard DeePC cost function up to a constant, and both problems yield the same optimal solution.
\end{theorem}

\begin{proof}
Since the scenario output constraints are inactive by assumption, they impose no binding restriction on the optimization in~\eqref{eq:Scenario-DeePC}, while the data-driven equality constraint and the input constraint are independent of the sampled scenarios $\{s^{(i)}\}_{i=1}^{N_{\mathrm{scen}}}$. The only term in~\eqref{eq:Scenario-DeePC} that depends on the scenarios is the averaged tracking cost, which expands as
\begin{equation}
\label{eq:proof_thm2_expansion}
\resizebox{0.95\linewidth}{!}{$
\displaystyle
\begin{multlined}
\hfill \frac{1}{N_{\mathrm{scen}}} \sum_{i=1}^{N_{\mathrm{scen}}}
\| r - y_{\mathrm{pred}} - s^{(i)} \|_Q^2 = \hfill \\
\| r - y_{\mathrm{pred}} \|_Q^2
- 2\, (r - y_{\mathrm{pred}})^\top Q \,\bar s_{N_{\mathrm{scen}}}
+ \frac{1}{N_{\mathrm{scen}}} \sum_{i=1}^{N_{\mathrm{scen}}} \| s^{(i)} \|_Q^2,
\end{multlined}
$}
\end{equation}
where
\begin{equation}
\label{eq:proof_thm2_sample_mean}
\bar s_{N_{\mathrm{scen}}}
\;:=\; \frac{1}{N_{\mathrm{scen}}} \sum_{i=1}^{N_{\mathrm{scen}}} s^{(i)}.
\end{equation}

The per-step blocks of the scenarios in~\eqref{eq:scenario_sampling} are drawn i.i.d.\ uniformly from the finite buffer $\mathcal{E}_{\mathrm{buffer}}$, so by the strong law of large numbers (SLLN) the sample mean converges block-wise, as $N_{\mathrm{scen}} \to \infty$, to the buffer mean:
\begin{align}
\label{eq:proof_thm2_LLN}
\bar s_{N_{\mathrm{scen}}}
\;&\xrightarrow{N_{\mathrm{scen}} \to \infty}\;
\mathbf{1}_N \otimes \frac{1}{N_{\mathrm{buffer}}} \sum_{j=1}^{N_{\mathrm{buffer}}} w^{(j)}\\
&\xrightarrow{N_{\mathrm{buffer}} \to \infty}\;
\mathbf{1}_N \otimes \bar w \;=\; 0,
\end{align}
where the second limit is Assumption~\ref{ass:zero_mean_process}. Hence the cross term in~\eqref{eq:proof_thm2_expansion} vanishes only in this iterated limit.

The third term in~\eqref{eq:proof_thm2_expansion} converges, again by the SLLN, to the second moment of the prediction error process:
\begin{equation}
\label{eq:proof_thm2_cw}
\frac{1}{N_{\mathrm{scen}}} \sum_{i=1}^{N_{\mathrm{scen}}} \| s^{(i)} \|_Q^2
\;\xrightarrow{N_{\mathrm{scen}} \to \infty}\;
c_s \;\ge\; 0.
\end{equation}
Crucially, $c_s$ depends only on the prediction error process and is constant with respect to the decision variables $(g, u, y_{\mathrm{pred}}, \sigma_y)$.
Combining \eqref{eq:proof_thm2_expansion}--\eqref{eq:proof_thm2_cw}, the Scenario-DeePC cost converges to
\begin{equation}
\label{eq:proof_thm2_limit_cost}
\| r - y_{\mathrm{pred}} \|_Q^2
+ \|u\|_R^2 + \|g\|_{\lambda_g}^2 + \|\sigma_y\|_{\lambda_y}^2
+ c_s.
\end{equation}
Since $c_s$ is independent of the decision variables, removing it does not change the minimizer. The minimization of~\eqref{eq:proof_thm2_limit_cost} therefore coincides with the optimization problem of standard DeePC~\eqref{eq:DeePC}. The optimal solutions of the two problems agree. 
\end{proof}

\begin{assumption}[Biased prediction error process]
\label{ass:biased_process}
The prediction error process $\{w^{(j)}\}$ generating the scenario buffer $\mathcal{E}_{\mathrm{buffer}}$, defined in~\eqref{eq:prediction_error} and capturing the combined effect of external disturbances, measurement noise, and residual data-driven modeling inaccuracies, has nonzero mean in the sense that, as $N_{\mathrm{buffer}} \to \infty$, 
\begin{equation}
\label{eq:biased_process}
\bar w \;:=\; \lim_{N_{\mathrm{buffer}} \to \infty}
\frac{1}{N_{\mathrm{buffer}}} \sum_{j=1}^{N_{\mathrm{buffer}}} w^{(j)}
\;\neq\; 0.
\end{equation}
\end{assumption}

\begin{theorem}[Scenario-DeePC as offset-corrected DeePC]
\label{thm:offset_correction}
Let Assumptions~\ref{ass:interior} and~\ref{ass:biased_process} hold. Then, in the iterated limit $N_{\mathrm{scen}} \to \infty$ followed by $N_{\mathrm{buffer}} \to \infty$, the Scenario-DeePC problem~\eqref{eq:Scenario-DeePC} is equivalent to a standard DeePC problem with shifted reference
\begin{equation}
\label{eq:shifted_reference}
\tilde r \;=\; r - \bar s, \qquad \bar s := \mathbf{1}_N \otimes \bar w,
\end{equation}
where $\bar s$ replicates the per-step bias $\bar w$ across the prediction horizon. In particular, the optimal predicted output satisfies $y_{\mathrm{pred}}^\star \to r - \bar s$, so that the mean scenario output satisfies $y_{\mathrm{pred}}^\star + \bar s \to r$. Scenario-DeePC thus automatically compensates the systematic prediction bias $\bar w$.
\end{theorem}

\begin{proof}
Repeating the expansion~\eqref{eq:proof_thm2_expansion} and the SLLN argument~\eqref{eq:proof_thm2_LLN}, but now with $\bar w \neq 0$ by Assumption~\ref{ass:biased_process}, the cross term no longer vanishes: in the iterated limit $N_{\mathrm{scen}} \to \infty$ then $N_{\mathrm{buffer}} \to \infty$ the sample mean tends to $\bar s$, and the Scenario-DeePC cost converges to 
\begin{equation}
\label{eq:proof_thm3_limit}
\scalebox{0.85}{$
\displaystyle
\| r - y_{\mathrm{pred}} \|_Q^2
- 2\, (r - y_{\mathrm{pred}})^\top Q\, \bar s
+ c_s
+ \|u\|_R^2 + \|g\|_{\lambda_g}^2 + \|\sigma_y\|_{\lambda_y}^2,
$}
\end{equation}
where $c_s$ is the second moment of the prediction error process, defined as in~\eqref{eq:proof_thm2_cw}.

Completing the square in the tracking part:
\begin{equation}
\label{eq:proof_thm3_complete_square}
\begin{multlined}
\hfill\| r - y_{\mathrm{pred}} \|_Q^2 - 2 (r - y_{\mathrm{pred}})^\top Q \bar s = \\
\| (r - \bar s) - y_{\mathrm{pred}} \|_Q^2 - \| \bar s \|_Q^2.\hfill
\end{multlined}
\end{equation}
Substituting back into~\eqref{eq:proof_thm3_limit}, the Scenario-DeePC cost converges to
\begin{equation}
\label{eq:proof_thm3_shifted_cost}
\| \tilde r - y_{\mathrm{pred}} \|_Q^2
+ \|u\|_R^2 + \|g\|_{\lambda_g}^2 + \|\sigma_y\|_{\lambda_y}^2 + c',
\end{equation}
where $\tilde r = r - \bar s$ and $c' = c_s - \|\bar s\|_Q^2$ is constant with respect to the decision variables.
\eqref{eq:proof_thm3_shifted_cost} is precisely the cost of the standard DeePC problem~\eqref{eq:DeePC} with reference $\tilde r$ in place of $r$, plus an optimization-independent constant. Since $c'$ does not influence the minimizer, the optimal predicted output satisfies $y_{\mathrm{pred}}^\star \to \tilde r = r - \bar s$, and consequently
\begin{equation}
\label{eq:proof_thm3_offset_compensation}
y_{\mathrm{pred}}^\star + \bar s \;\to\; r,
\end{equation}
i.e., Scenario-DeePC steers the expected output toward the true reference.
\end{proof}
Both limits $N_{\mathrm{scen}} \to \infty$ and $N_{\mathrm{buffer}} \to \infty$ are mathematical devices that invoke the SLLN to obtain the equivalences of Theorems~\ref{thm:zero_mean_equivalence} and~\ref{thm:offset_correction} in closed form. In practice, both are finite and bounded by storage. The results nevertheless capture the behavior of Scenario-DeePC for sufficiently large buffers and scenario counts: under interior operation, its cost is well approximated by that of standard DeePC with the original reference when the prediction error is zero-mean, and with the bias-corrected reference $\tilde r = r - \bar s$ when it is not.

\section{Numerical Results}
\label{sec:Num_Results}

In this section, the performance of the proposed Scenario-DeePC approach is evaluated on two representative systems. First, an LTI model of a Boeing 747 is considered to demonstrate the effectiveness of the method in a setting where the assumptions underlying behavioral system theory are satisfied. Second, a nonlinear two-tank system is used to assess the performance of Scenario-DeePC beyond the LTI regime and to illustrate its enhanced robustness with respect to nonlinear dynamics.

Given the parameters in Table~\ref{tab:HP} and~\eqref{eq:nopt}, we obtain $n_{\mathrm{opt,Boeing}} = 961$ and $n_{\mathrm{opt,2\text{-}Tank}} = 176$. For $\beta = 10^{-6}$ and $\varepsilon = 0.1$, the corresponding lower bounds on the required number of scenarios according to~\eqref{eq:sample_complexity} are $N_{\mathrm{scen,Boeing}} = 19497$ and $N_{\mathrm{scen,2\text{-}Tank}} = 3797$.
As outlined in Section~\ref{sec:Theo_Results}, these bounds can be overly conservative in practice. In the following, we demonstrate through numerical simulations that satisfactory constraint satisfaction can be achieved with significantly fewer scenarios, as reported in Table~\ref{tab:HP}.

\subsection{Example 1: LTI Boeing 747}
As a first case study, we consider a benchmark example based on the longitudinal flight dynamics of a Boeing 747 from~\cite{Camacho_2007}. The continuous-time model is discretized with a sampling time of $T_s = 0.1\,\text{s}$, resulting in a discrete-time LTI system of the form~\eqref{eq:system} with the following system matrices (rounded to two decimal places):
{\setlength{\arraycolsep}{2pt}
\begin{equation}
\scalebox{0.89}{$
\begin{aligned}
&A =
\begin{bmatrix}
1.00 & 0.00 & -0.00 & -0.03 \\
-0.01 & 0.96 & 0.74 & 0.00 \\
0.00 & -0.01 & 0.95 & -0.00 \\
0.00 & -0.00 & 0.10 & 1.00
\end{bmatrix}, \;
B =
\begin{bmatrix}
0.00 & 0.10 \\
-0.06 & 0.02 \\
-0.11 & 0.06 \\
-0.01 & 0.00
\end{bmatrix}, \\
& C =
\begin{bmatrix}
1 & 0 & 0 & 0 \\
0 & -1 & 0 & 7.74
\end{bmatrix}, \quad
D =
\begin{bmatrix}
0 & 0 \\
0 & 0
\end{bmatrix}.
\end{aligned}$}
\end{equation}
}
The two controlled outputs are the longitudinal velocity $y_1$ and the climb rate $y_2$, which are controlled via the elevator $u_1$ and the throttle $u_2$. For a more detailed explanation of the simulated system, we refer to Chapter 6.4 of~\cite{Camacho_2007}. Input and output constraints are imposed to reflect physical and operational limitations of the system according to~\cite{Lazar_2023}:
\begin{equation}
\begin{aligned}
&\mathcal{U} := \{ u \in \mathbb{R}^2 : -20 \leq u_1 \leq 20,\; -20 \leq u_2 \leq 20\}, \\
&\mathcal{Y} := \{ y \in \mathbb{R}^2 : -25 \leq y_1 \leq 25,\; -15 \leq y_2 \leq 15 \}.
\end{aligned}
\end{equation}

Both DeePC and Scenario-DeePC use the same data within the Hankel matrices with the same hyperparameters (HP) displayed in Table~\ref{tab:HP}, which proved a robust choice in the HP tuning process.
Measurement noise during both scenario generation and control operation is drawn independently for each output channel according to
\begin{equation}
v_{\ell} \sim \mathcal{N}\left(0,\,(0.01\,y_{\max,\ell})^2\right), \quad \ell=1,2,
\label{eq:noise_sampling}
\end{equation}
which corresponds to a sensor inaccuracy of $1\%$ of the maximum output value. No external disturbances are considered in this case.
To verify Assumption~\ref{ass:mixing} empirically, we compute the sample autocorrelation function (ACF) of the buffered prediction errors,
\begin{equation}
    \rho(\tau)
    \;=\;
    \frac{
        \sum_{j=1}^{N_\mathrm{buffer}-\tau}
        \bigl(w^{(j)} - \bar{w}\bigr)
        \bigl(w^{(j+\tau)} - \bar{w}\bigr)
    }{
        \sum_{j=1}^{N_\mathrm{buffer}}
        \bigl(w^{(j)} - \bar{w}\bigr)^2
    },
    \label{eq:empirical_acf}
\end{equation}
where $\tau$ is the lag and $\bar{w}$ is the sample mean of the buffer. For the offline buffer used above, the empirical ACF already decays to within the 95\,\% confidence band $\pm 1.96/\sqrt{N_\mathrm{buffer}}$ for $\tau \geq M_\mathrm{emp}$ with $M_\mathrm{emp} = 0$ (see Figure~\ref{fig:acf}), so the buffered prediction errors are effectively independent and every computed prediction error can be retained, in line with Assumption~\ref{ass:mixing}.

\begin{figure}
    \centering
    \includegraphics[width=\linewidth]{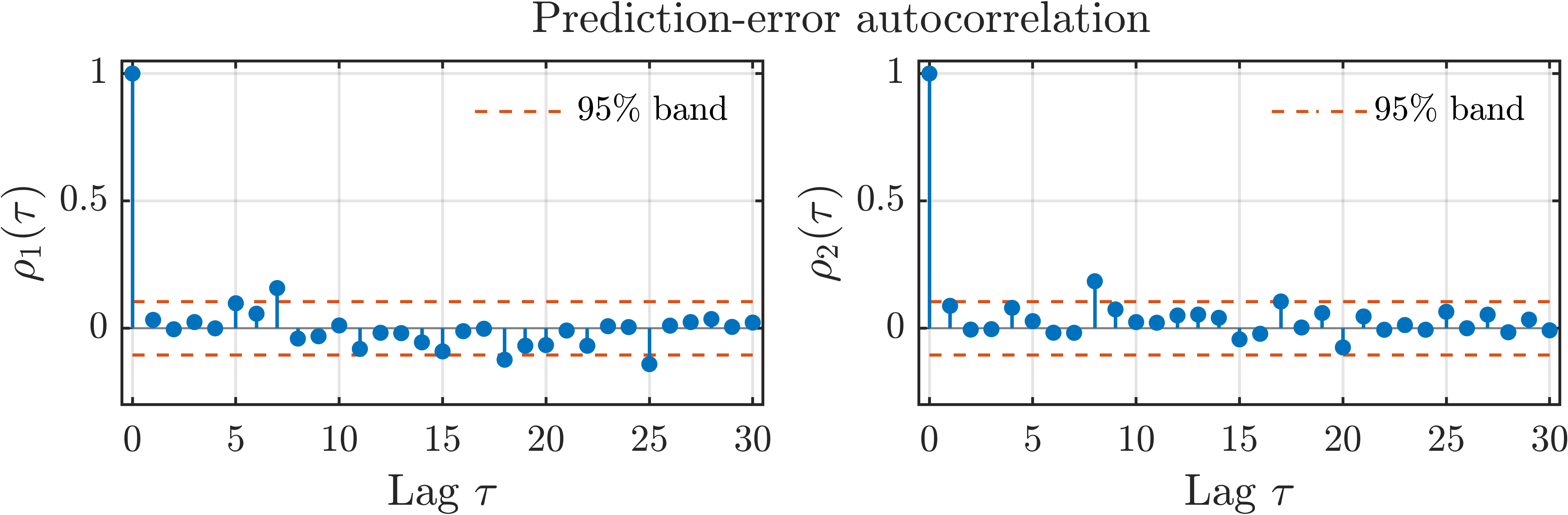}
    \caption{Prediction-error autocorrelation $\rho_\ell(\tau)$ for the Boeing~747 outputs $y_1$ (left) and $y_2$ (right), with $95\%$ white-noise band.}
    \label{fig:acf}
\end{figure}
\subsubsection{Nominal Case}
In the nominal case, the system is operated around an operating point (OP) sufficiently far from the imposed constraints. This setting is chosen to evaluate the reference tracking performance of DeePC and Scenario-DeePC without the influence of active constraint handling. 
As illustrated in the first 200 time steps of Figure~\ref{fig:Boeing}, DeePC and Scenario-DeePC exhibit comparable reference tracking performance when subjected to identical measurement noise realizations in every time step during control operation. This observation is further supported by the RMSE values reported in Table~\ref{tab:RMSE}, which show similar performance for both output variables.
\begin{figure}
    \centering
    \includegraphics[width=.9\linewidth]{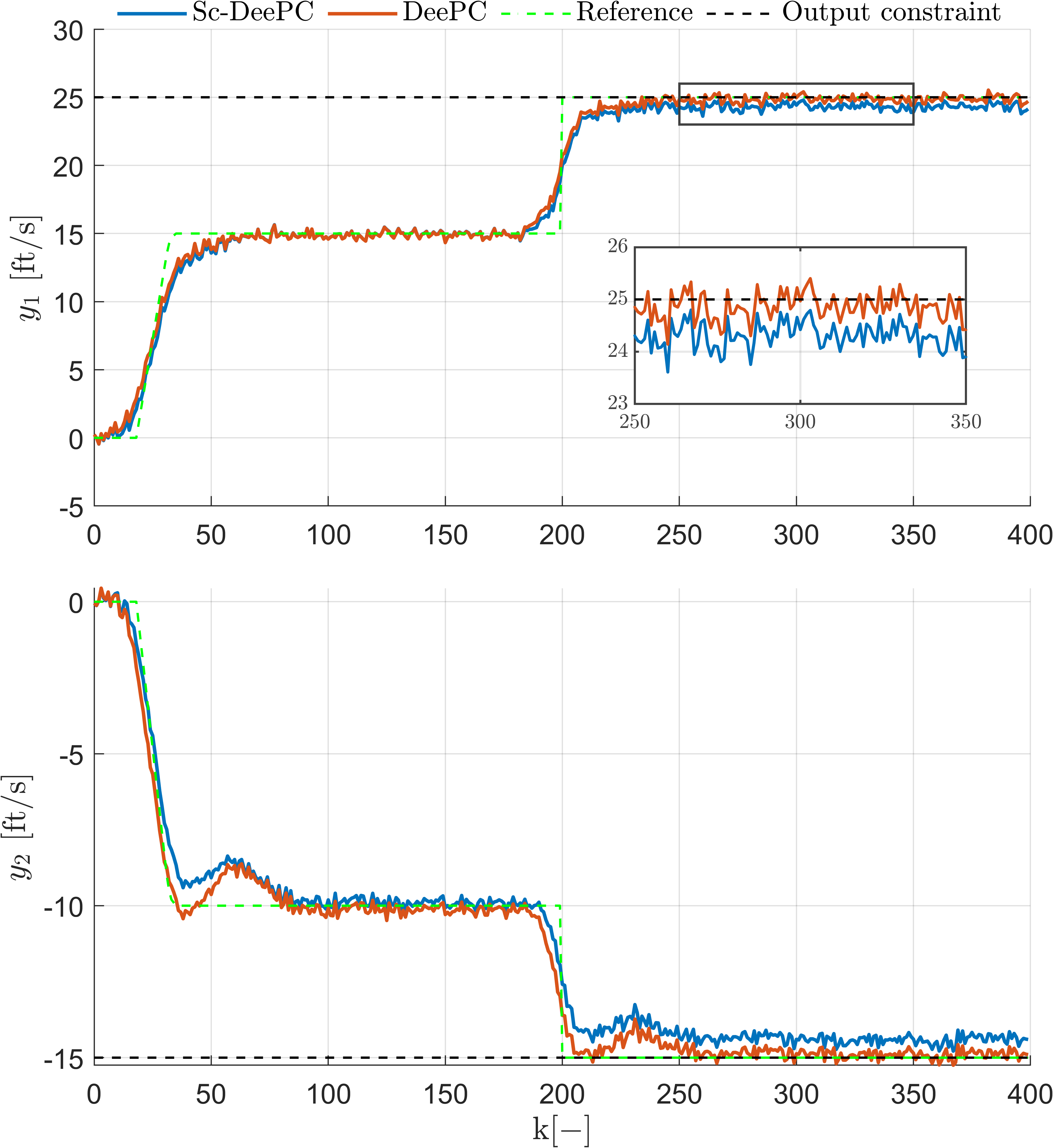}
    \caption{Nominal and robust case of the longitudinal system model controlled by Scenario-DeePC and DeePC with the velocity $y_1$ and the climb rate $y_2$ in $\mathrm{[ft/s]}$ for the discrete time steps $k$.  }
    \label{fig:Boeing}
\end{figure}
\subsubsection{Robust Case}
In the robust case, the system is operated on the edge of the constraint set to provoke constraint violations. This setup is chosen to evaluate the ability of Scenario-DeePC to enforce constraint satisfaction under uncertainty and to compare it with the standard DeePC formulation.
As shown in the last 200 time steps of Figure~\ref{fig:Boeing}, DeePC exhibits frequent constraint violations due to the absence of an explicit mechanism to account for uncertainty in constraint enforcement. In contrast, Scenario-DeePC significantly reduces the number of violations by incorporating chance constraints through the scenario-based formulation.
This improved constraint satisfaction comes at the cost of increased conservatism, as the scenario-based constraints tighten the admissible set and thereby degrade reference tracking performance compared to standard DeePC, as reflected in the evaluation metrics in Table~\ref{tab:RMSE}.
\begin{table}[]
\caption{RMSE and number of constraint violations for both Scenario-DeePC and DeePC for all experimental setups.}
\resizebox{\columnwidth}{!}{%
\begin{tabular}{l|cccc|cc}
            & \multicolumn{4}{c|}{$\mathrm{RMSE}$}                                                         & \multicolumn{2}{c}{$n_{\mathrm{violation}}$} \\
            & \multicolumn{2}{c|}{Scenario-DeePC}                           & \multicolumn{2}{c|}{DeePC}         & \multicolumn{1}{c|}{Scenario-DeePC}     & DeePC    \\
            & \multicolumn{1}{c|}{$y_1$} & \multicolumn{1}{c|}{$y_2$} & \multicolumn{1}{c|}{$y_1$} & $y_2$ & \multicolumn{1}{c|}{}             &          \\ \hline
Nominal     & 1.04 & \multicolumn{1}{c|}{0.65} & 1.04 & 0.68 & \multicolumn{1}{c|}{0}  & 0   \\
Robust      & 1.08 & \multicolumn{1}{c|}{0.80} & 0.74 & 0.35 & \multicolumn{1}{c|}{0}  & 82  \\
Adaptive    & 2.46 & \multicolumn{1}{c|}{2.02} & 2.27 & 1.95 & \multicolumn{1}{c|}{11} & 97  \\
2-Tank Nom. & 0.44 & \multicolumn{1}{c|}{-}    & 0.97 & -    & \multicolumn{1}{c|}{0}  & 0   \\
2-Tank Rob. & 0.66 & \multicolumn{1}{c|}{-}    & 2.71 & -    & \multicolumn{1}{c|}{40} & 397
\end{tabular}%
}
\label{tab:RMSE}
\end{table}

Figure~\ref{fig:Violations_RMSE} illustrates the RMSE of the outputs $y_1$ and $y_2$ with respect to their references, as well as the total number of constraint violations over the 200 time steps of the robust case. Both metrics are evaluated as a function of the number of scenarios $N_{\mathrm{scen}}$ used within the Scenario-DeePC framework, where $N_{\mathrm{scen}} = 0$ corresponds to standard DeePC.
Increasing the number of scenarios reduces constraint violations while gradually raising the RMSE, as the tightened constraints prevent the system from reaching references on the boundary of the admissible set. The trade-off is favorable: violations drop drastically for even a few scenarios, at a comparatively minor tracking cost increase, highlighting the effectiveness of the scenario-based approach.
\begin{figure}
    \centering
    \includegraphics[width=.75\linewidth]{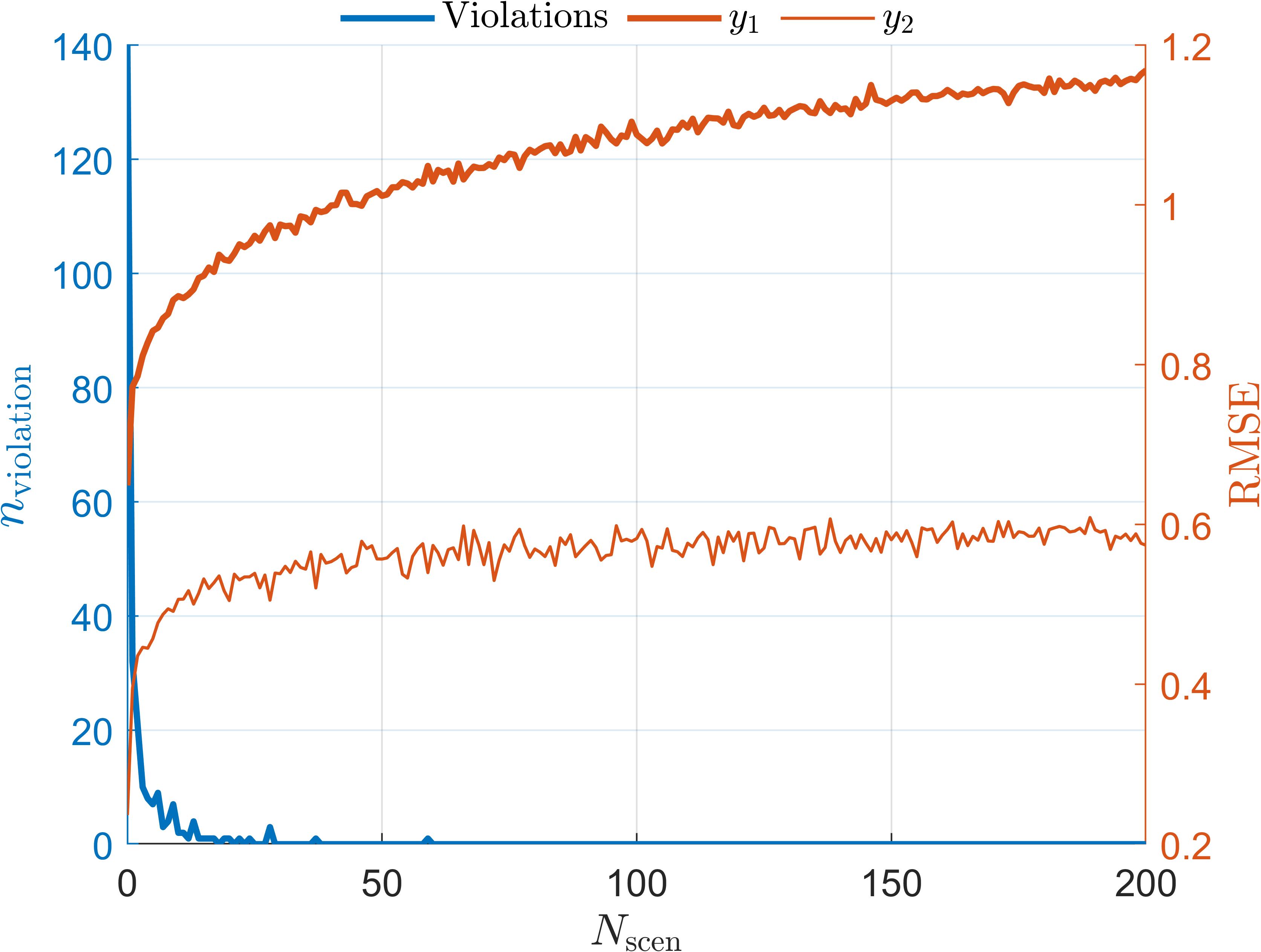}
    \caption{Number of constraint violations $n_{\mathrm{violation}}$ and RMSE of both outputs $y_1$ and $y_2$ dependent on the number of scenarios $N_{\mathrm{scen}}$ chosen for Scenario-DeePC.}
    \label{fig:Violations_RMSE}
\end{figure}

\subsubsection{Adaptive Scenario-DeePC}
In the adaptive Scenario-DeePC setting for the LTI Boeing model, the scenario replay buffer is not initialized with prior data but is instead populated online during operation. At each time step, new scenarios are generated according to \eqref{eq:prediction_error} and added to the buffer, while the oldest scenarios are discarded.
This setup is used to demonstrate the ability of Scenario-DeePC to adapt its constraint tightening in response to changing measurement noise and disturbances. Using the same Gaussian model as above, the noise standard deviation is set to $2\%$ of $y_{\max,\ell}$ initially, ramped down to $0.5\%$ at $k=100$ and back up to $1\%$ at $k=200$, while at $k=300$ a non-zero-mean disturbance reaching a constant offset of $1.0\,\mathrm{ft/s}$ on $y_1$ (representing a wind disturbance) is introduced. As discussed previously, both DeePC and Scenario-DeePC exhibit comparable tracking performance for LTI systems in the nominal case. Therefore, only the robust case is presented in Figure~\ref{fig:adaptive}, where differences in control performance become apparent.
At the beginning of the simulation, Scenario-DeePC behaves similarly to standard DeePC, as the scenario replay buffer is not yet sufficiently populated. Once enough scenarios have been collected, Scenario-DeePC exhibits increased constraint tightening, leading to improved constraint satisfaction compared to DeePC.
When the magnitude of the measurement noise is reduced at $k=100$, the scenario replay buffer gradually adapts, resulting in a reduction of conservatism over time. Conversely, when the noise level increases again at $k=200$, Scenario-DeePC temporarily exhibits a higher number of constraint violations until the scenario replay buffer is updated sufficiently, after which the constraint tightening increases accordingly.
Similarly, the introduction of a non-zero-mean disturbance on output $y_1$ at $k=300$, representing external influences such as wind, leads to a temporary mismatch until the replay buffer reflects the new operating conditions. Once updated, Scenario-DeePC again adapts and restores improved constraint satisfaction.
The size of the scenario replay buffer directly influences the adaptation speed of the constraint tightening: larger buffers require more time until they are sufficiently updated to reflect changes in the operating conditions. At the same time, the buffer length is lower bounded by the number of scenarios, i.e., $N_{\mathrm{scen}} \leq N_{\mathrm{buffer}}$. For the setup of Figure~\ref{fig:adaptive}, a buffer length of $N_{\mathrm{buffer}} = 50$ is chosen to enable sufficiently fast adaptation while still providing a clear visualization of the constraint tightening effect.
\begin{figure}
    \centering
    \includegraphics[width=.9\linewidth]{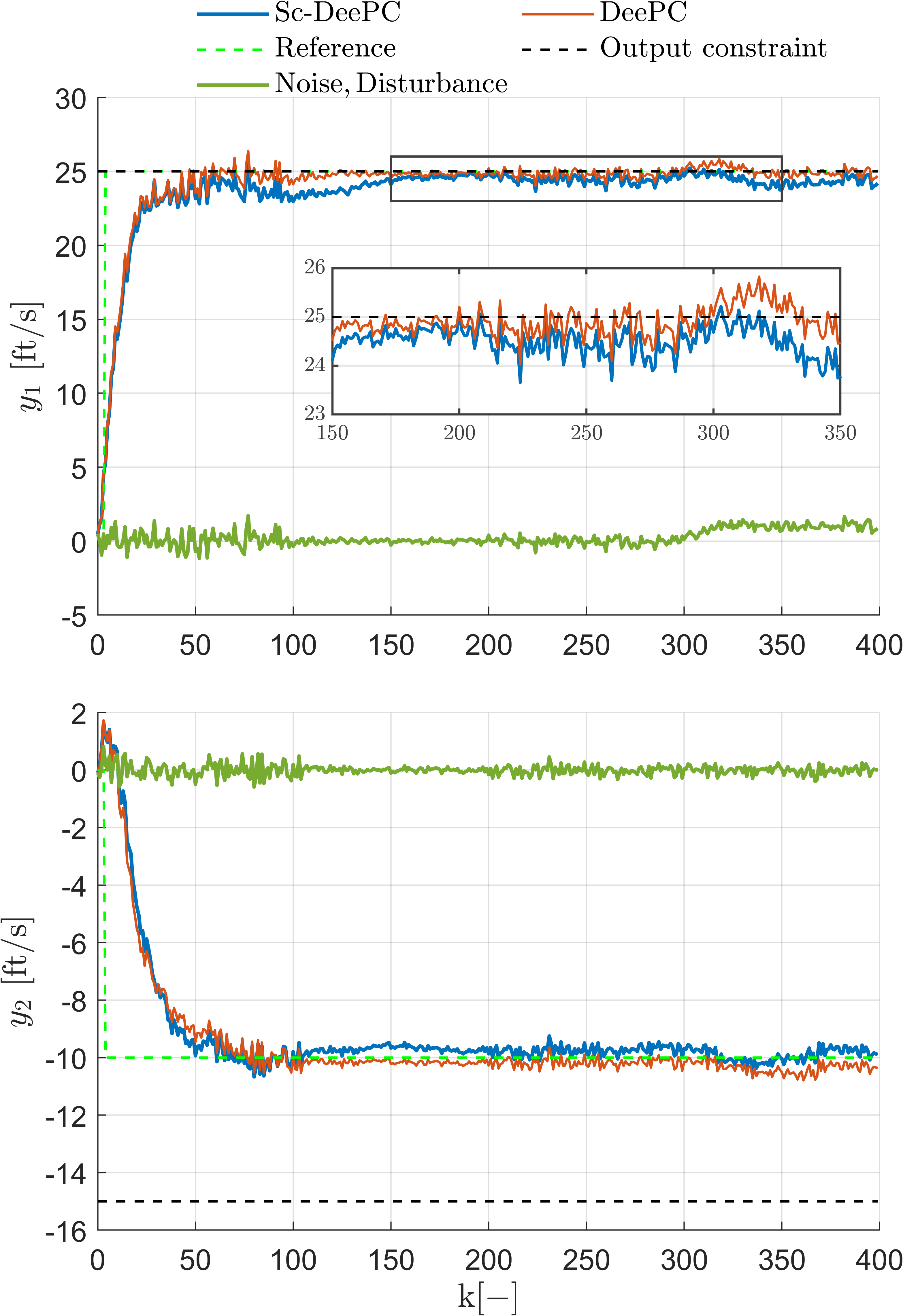}
    \caption{Adaptive Scenario-DeePC of the LTI Boeing 747 to visualize constraint tightening for changing external conditions. Direct comparison to DeePC under the same noise and disturbance profile (green).}
    \label{fig:adaptive}
\end{figure}
Overall, across all considered noise and disturbance configurations, adaptive Scenario-DeePC achieves significantly fewer constraint violations while maintaining good tracking performance. The corresponding RMSE values and the number of constraint violations are reported in Table~\ref{tab:RMSE}.
\begin{table}[]
\caption{Hyperparameters employed in the experiments for DeePC (excluding $N_{\mathrm{scen}}$, $N_{\mathrm{buffer}}$) and Scenario-DeePC.}
\resizebox{\columnwidth}{!}{%
\begin{tabular}{l|c|c|c|c|c|c|c|c|c}
         & $N$ & $T_{\text{ini}}$ & $T_d$  & $Q$ & $R$  & $\lambda_{y}$       & $\lambda_g$           & $N_{\mathrm{scen}}$ & $N_{\mathrm{buffer}}$ \\ \hline
Nominal, Robust & 20   & 20     & $10^{3}$       & $10^2$ & $10^{-3}$ & $10^7$ & $10^5$ & 50       & 400         \\
Adaptive & 20   & 20     & $10^{3}$       & $10^2$ & $10^{-3}$ & $10^7$ & $10^5$ & 25       & 50         \\
2-Tank    & 5   & 20   & $200$         & $10^{4}$ & $10^{-2}$ & $10^7$ & $10^4$ & 20       & 40           \\
\end{tabular}%
}
\label{tab:HP}
\end{table}
\subsection{Two-Tank System}
As a second case study, a nonlinear two-tank system is considered to evaluate the performance of Scenario-DeePC beyond the LTI setting. The system is represented as a Linear Parameter-Varying (LPV) model in~\eqref{eq:LPV_system}, where the scheduling variable $\theta$ depends on the water levels of the tanks according to~\eqref{eq:theta}. The system description is adopted from~\cite{Teppa_2025}, while the hyperparameters and sampling time $T_s =0.69\mathrm{s}$ are chosen consistently with~\cite{Zieglmeier_2025_Semi} and are listed in Table~\ref{tab:HP}. The constraint sets are given by~\eqref{eq:constraint}.
{\setlength{\arraycolsep}{1pt}
\begin{equation}
\scalebox{0.93}{$
\begin{aligned}
\dot{x}(t) &= 
\left[\begin{array}{cc}
-0.904 \theta_1(t) & 0 \\
0.904 \theta_1(t) & -0.508 \theta_2(t)
\end{array}\right] x(t)
+
\left[\begin{array}{c}
0.258 \\
0
\end{array}\right] u(t) \\
y(t) &= 
\left[\begin{array}{ll}
0 & 1
\end{array}\right] x(t).
\end{aligned}
$}
\label{eq:LPV_system}
\end{equation}}
\begin{equation}
{\theta}(t) =
\begin{bmatrix}
\theta_1(t) \\
\theta_2(t)
\end{bmatrix}
=
\begin{bmatrix}
1/\sqrt{x_1(t)} \\
1/\sqrt{x_2(t)}
\end{bmatrix}.
\label{eq:theta}
\end{equation}
\begin{equation}
\begin{aligned}
&\mathcal{U} := \{ u \in \mathbb{R}: 0 \leq u \leq 22 \}, \\
&\mathcal{Y} := \{ y \in \mathbb{R} : 0 \leq y \leq 25 \}.
\end{aligned}
\label{eq:constraint}
\end{equation}
The same Gaussian measurement-noise model of~\eqref{eq:noise_sampling} is applied to its single output ($n_y=1$) at $1\%$ of $y_{\max}$, and no external disturbance is injected.
Due to the nonlinear nature of the system, the data-driven model underlying DeePC cannot perfectly capture the system dynamics. As a result, standard DeePC exhibits a systematic OP-dependent prediction error, which manifests as an offset between predicted and measured outputs in addition to measurement noise. Regularization alone is insufficient to fully compensate for this effect.
In contrast, Scenario-DeePC leverages the online scenario generation mechanism to capture this mismatch. By continuously updating the scenario replay buffer with newly observed prediction errors, the method incorporates both stochastic disturbances and systematic data-driven model inaccuracies arising within the implicit identification step. As in the Boeing example, the decorrelation horizon is determined from the sample ACF~\eqref{eq:empirical_acf} of the standard-DeePC prediction error, evaluated per OP to remove the systematic offset. This yields $M=2$ here, so every third prediction error is added to the buffer, consistent with Assumption~\ref{ass:mixing}. Since these errors depend on the current OP, the scenario set adapts accordingly.
Through the inclusion of scenario-based outputs in the objective function, as defined in~\eqref{eq:Scenario-DeePC}, Scenario-DeePC effectively compensates for these offsets, established in Theorem~\ref{thm:offset_correction}, where a nonzero buffer mean acts as an effective shift of the reference.
Figure~\ref{fig:LPV} additionally shows the moving averages over 20 time steps of the outputs $y$ and the relative errors $e_{\mathrm{rel}}$, denoted by $y_{\mathrm{MA}}$ and $e_{\mathrm{rel,MA}}$, respectively, to highlight the offset caused by prediction inaccuracies that is otherwise hardly visible below the measurement noise. The first 200 time steps are omitted, as they are used to initialize the scenario replay buffer and to allow both control approaches to settle at the initial reference set point.
For the nominal cases ($k<600$) in Figure~\ref{fig:LPV}, once the replay buffer is sufficiently populated ($k\gtrapprox400$), Scenario-DeePC compensates the systematic offset and improves tracking, as seen in the decreasing moving average of the relative error, whereas standard DeePC cannot eliminate the offset arising from the nonlinear data in the Hankel matrices.
In the robust case ($k > 600$), where the system operates close to the constraint boundaries, Scenario-DeePC achieves significantly improved constraint satisfaction. In contrast, standard DeePC is unable to satisfy the constraints due to prediction errors arising from the system nonlinearities. At the same time, the relative error indicates that Scenario-DeePC leads to superior tracking performance compared to standard DeePC, further highlighting the benefit of the proposed approach for nonlinear systems.
\begin{figure}
    \centering
    \includegraphics[width=\linewidth]{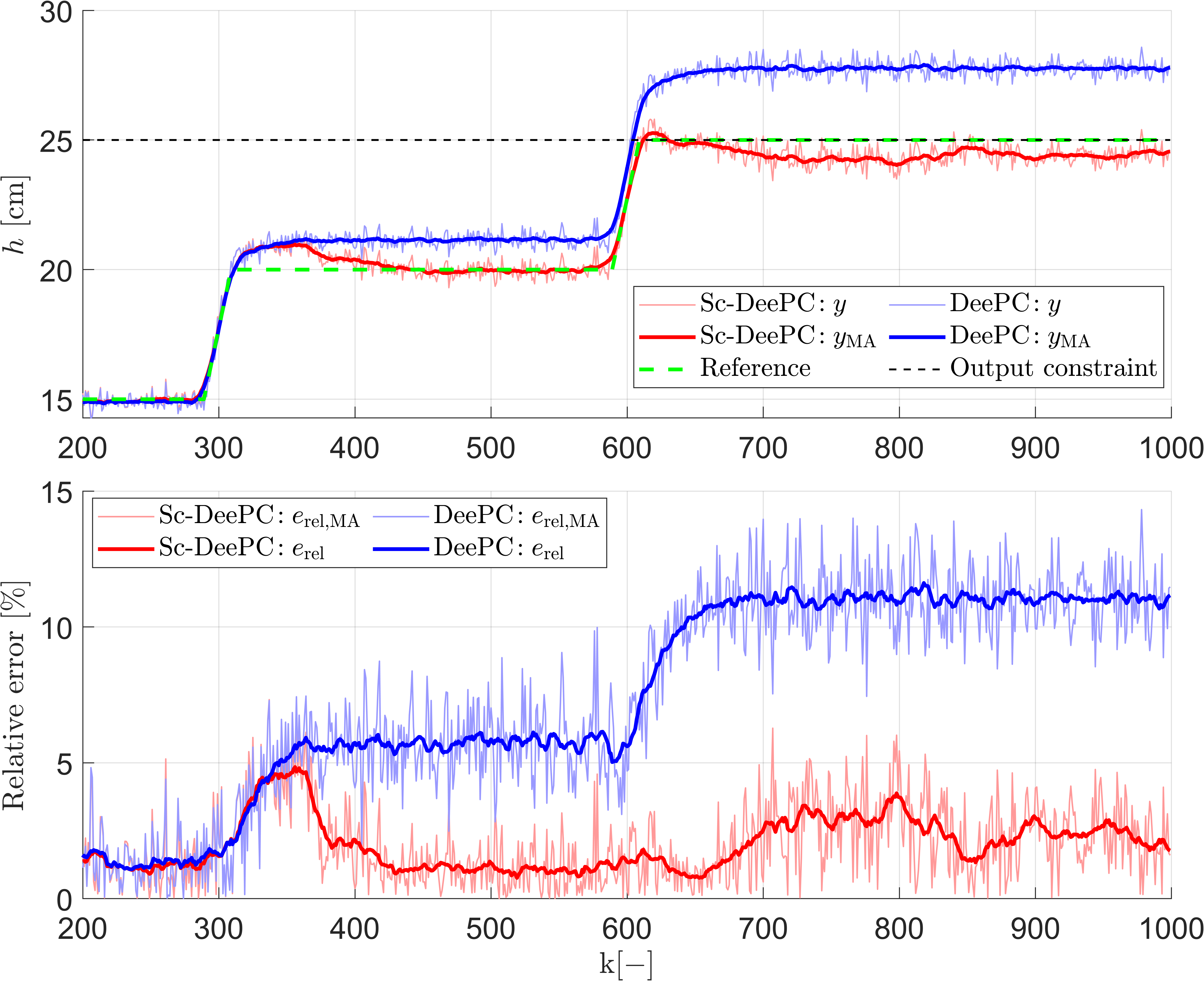}
    \caption{Upper: Water level $h$ of the second tank for standard DeePC and Scenario-DeePC tracking the reference trajectory. The thicker lines indicate the moving average of the respective outputs over a window of 20 time steps. Lower: Relative tracking error of both control approaches, including a moving average over 20 time steps (thicker lines).}
    \label{fig:LPV}
\end{figure}

\begin{remark}
    Since the prediction error offset depends on the current OP, effective compensation relies on the scenarios stored in the replay buffer accurately representing the current offset, which are continuously updated over time. For rapidly changing reference trajectories, the collected scenarios may not capture the OP-dependent offset sufficiently well, which can lead to reduced tracking performance compared to the steady set-point cases shown in Figure~\ref{fig:LPV}.
\end{remark}

\section{Conclusion}
\label{sec:Conclusion}
We introduced Scenario-DeePC, which combines DeePC with the scenario approach by leveraging an empirically collected distribution of prediction errors, fully aligning with the data-driven philosophy. The proposed method was developed to improve constraint satisfaction under uncertainty. On the theoretical side, we established well-posedness and the decision dimension of the scenario program, a probabilistic constraint-satisfaction guarantee, recursive feasibility, and consistency with standard DeePC, including an offset-compensation result (Theorem~\ref{thm:offset_correction}) in which a biased prediction-error buffer acts as an effective reference shift.
The numerical results demonstrate that Scenario-DeePC significantly reduces constraint violations compared to standard DeePC when operating near constraint boundaries, while maintaining comparable reference tracking performance in nominal conditions. Furthermore, the adaptive Scenario-DeePC formulation enables the controller to adjust its constraint tightening in response to changing noise and disturbance characteristics.
For the nonlinear two-tank system, Scenario-DeePC additionally improves tracking performance by compensating for OP-dependent prediction errors captured through the continuously updated empirical scenario replay buffer. Overall, Scenario-DeePC provides a simple and effective data-driven framework for achieving robust constraint satisfaction in control.

On the theoretical side, future work targets stability guarantees with data-based terminal ingredients or probabilistic reachable sets in the LTI setting. On the practical side, the OP-dependent prediction errors of nonlinear systems motivate multiple OP-dependent scenario buffers, combined with Gain-Scheduled-DeePC~\cite{Zieglmeier_2025_GS} for a more practical setup.

\hfill

\begingroup
\scriptsize
\noindent\textbf{Declaration of Generative AI and AI Technologies in the Writing Process:}\\
During the preparation of this work, the author(s) used ChatGPT in order to improve the
grammar of this paper. After using this tool/service, the author(s) reviewed and edited the
content as needed and take full responsibility for the content of the publication.

\noindent\textbf{Conflicts of Interest:}\\
The authors declare no conflicts of interest.

\noindent\textbf{Data Availability Statement:}\\
Code available under: \href{https://github.com/SebsDevLab/Scenario-DeePC.git}{https://github.com/SebsDevLab/Scenario-DeePC.git}
\par
\endgroup

\bibliography{Bib_Sc_DeePC}

\end{document}